\documentclass[12pt,a4paper]{article}
\usepackage{epsfig,amsmath,amssymb,latexsym,graphicx,setspace,fullpage}
\usepackage{soul}
\usepackage{natbib}
\usepackage{amsthm}
\usepackage{enumerate}
\usepackage{multirow}
\usepackage{algorithm}
\usepackage{algpseudocode}
\makeatletter
\def\BState{\State\hskip-\ALG@thistlm}
\algnewcommand{\LineComment}[1]{\Statex \hskip\ALG@thistlm#1}

\usepackage{xcolor}
\usepackage{subcaption}

\def\ALG@step%
   {%
   \addtocounter{ALG@line}{1}%
   \addtocounter{ALG@rem}{1}%
   \ifthenelse{\equal{\arabic{ALG@rem}}{\ALG@numberfreq}}%
      {\setcounter{ALG@rem}{0}\alglinenumber{\thealgorithm.\arabic{ALG@line}}}
      {}%
   }%
 
   \makeatother
   
\newcommand{\tr}{\mbox{tr}}

\newcommand{\bhat}{\hat{\beta}}

\newcommand{\ataui}{a_{\tau}^{(i)}}
\newcommand{\atauione}{a_{\tau}^{(i-1)}}
\newcommand{\btaui}{b_{\tau}^{(i)}}
\newcommand{\btauione}{b_{\tau}^{(i-1)}}
\newcommand{\aji}{a_{j}^{(i)}}
\newcommand{\ajione}{a_{j}^{(i-1)}}
\newcommand{\bji}{b_{j}^{(i)}}
\newcommand{\bjione}{b_{j}^{(i-1)}}
\newcommand{\nui}{\nu^{(i)}}
\newcommand{\nuione}{\nu^{(i-1)}}
\newcommand{\Si}{S^{(i)}}
\newcommand{\Sione}{S^{(i-1)}}

\newcommand{\Omegajione}{\Omega_j^{(i-1)}}
\newcommand{\bhatji}{\hat{\beta}_j^{(i)}}
\newcommand{\bhatjione}{\hat{\beta}_j^{(i-1)}}

\newcommand{\mutaui}{\mu_{\tau^{-1}}^{(i)}}
\newcommand{\mutauione}{\mu_{\tau^{-1}}^{(i-1)}}
\newcommand{\musigmai}{\mu_{\Sigma^{-1}}^{(i)}}

\newcommand{\nuhat}{\hat{\nu}}
\newcommand{\Shat}{\hat{S}}

\newcommand{\Omegajhat}{\hat{\Omega}_j}
\newcommand{\atauhat}{\hat{a}_\tau}
\newcommand{\btauhat}{\hat{b}_\tau }

\newcommand{\bhatjn}{\hat{\beta}_j^{(n)}}
\newcommand{\Omegajn}{\Omega_j^{(n)}}
\newcommand{\mutaun}{\mu_{\tau^{-1}}^{(n)}}
\newcommand{\Sn}{S^{(n)}}
\newcommand{\nun}{\nu^{(n)}}

\newtheorem{lemma}{Lemma}

\begin{document}

{\renewcommand{\thefootnote}{\fnsymbol{footnote}}
\begin{center}
{\Large \bf Flexible online multivariate regression with variational Bayes and the matrix-variate
Dirichlet process \\}
\end{center}

\begin{center}
Meng Hwee Victor Ong, David J Nott
and Ajay Jasra\footnote{  
Victor Ong is Research Associate, Department of Statistics and Applied Probability, 
National University of Singapore, Singapore 117546.  (email : victor84@nus.edu.sg).  
David J. Nott is Associate Professor, Department of Statistics and Applied Probability,
National University of Singapore, Singapore 117546. (email :standj@nus.edu.sg).  
Ajay Jasra is Associate Professor, Department of Statistics and Applied Probability, 
National University of Singapore, Singapore 117546 (email : staja@nus.edu.sg).}
\footnote{Victor Ong, David J. Nott and Ajay Jasra were supported by a Singapore Ministry of Education Academic Research Fund Tier 2 grant (R-155-000-143-112).}
\end{center}

\vspace{3mm}

\centerline{SUMMARY} 

\vspace*{2mm}
\noindent
Flexible regression methods where interest centres on the way that the whole distribution of a response vector changes with covariates are very useful in some applications.  
A recently developed technique in this regard uses the matrix-variate Dirichlet process as a prior 
for a mixing distribution on a coefficient in a multivariate linear regression
model.  The method is attractive, particularly in the multivariate setting, for the convenient way that it allows for borrowing strength across different
component regressions and for its computational simplicity and tractability.  The purpose of the present article is to develop fast online variational Bayes approaches 
to fitting this model and to investigate how they perform compared to MCMC and batch variational methods in a number of scenarios.  

\vspace*{4mm}\noindent
\section{Introduction}
\label{sec:intro}

Flexible modelling of multivariate conditional densities is a fundamental problem in statistics, particularly in regression applications in which there is interest
in the ways that the whole distribution of a response vector depends on covariates.   In a recent paper \citet{Zhang2010} developed a flexible
multivariate regression method using a Dirichlet process
prior for a mixing distribution on the coefficient in a multivariate linear model, where the Dirichlet process base prior is a matrix-variate normal distribution.  
The approach is attractive for its flexibility, the easy way it allows borrowing of strength between regressions for different response variables through the matrix-variate normal
base prior, and the computational
simplicity and convenience that comes from basing the method on the ordinary Dirichlet process.   
They refer to the Dirichlet process prior with matrix-variate normal base measure as the matrix-variate Dirichlet process (hereafter MDP), and further applications beyond
the multivariate linear regression setup were considered in \citet{zhang+wdj14}.  
The contribution of the present work is to consider fast online approaches to fitting the model
of \citet{Zhang2010} using variational Bayes methods, suitable for application in the context of large datasets. We also consider a novel approach to improving the
predictive performance of the online algorithm which gives performance comparable in many cases to a batch variational Bayes or MCMC approach.

In Bayesian nonparametrics, the development of suitable prior distributions for regression problems of the kind we consider here, 
involves the development of dependent prior distributions for sets of distributions indexed by the covariates.  A recent survey 
on the extensive literature on this topic is given by \citet{foti+w15}.  A key early paper is by \citet{maceachern00}, who introduced the framework of the
dependent Dirichlet process and which inspired many later developments.  
Some of the existing approaches in the literature include starting from the stick breaking representation of a random measure
and allowing distribution atoms or weights to be covariate dependent \citep{deiorio+mrm04,gelfand+km05,griffin+s06,dunson+p08}; consideration of covariate dependent generalizations
of the Chinese restaurant process or P\'{o}lya urn prediction rule \citep{blei+f11,caron+dd07}; as well as methods that build on normalized completely random measures \citep{kingman67,lijoi+p10} and which use their relationship with Poisson processes to introduce covariate dependence in various ways \citep{rao+t09,chen+rbt13,lijoi+np14}.  The 
above list of references is by no means exhaustive.  For the 
special case of grouped data, the hierarchical Dirichlet process \citep{teh+jbb06} has also been an extremely 
important development.

As mentioned, in the present work we consider the model of \citet{Zhang2010} which is attractive in the case of multivariate response for the convenient mechanism
it represents for borrowing strength across regressions for different components through the matrix-variate normal base prior.  Our objective is to 
develop fast online variational Bayes methods which allow the model of \citet{Zhang2010} to be applied with large datasets.  The approach adopted builds on
the VSUGS algorithm of \citet{zhang+nyj14} for Dirichlet process mixture models, which is a variational extension of the SUGS algorithm of \citet{wang+d11}.  
\citet{lin13} independently developed a similar algorithm to that of \citet{zhang+nyj14}.
The development of fast variational methods
for complex Bayesian nonparametric models has been a very active area of recent research, with an important early paper being \citet{blei+j06} where a batch
variational algorithm for fitting Dirichlet process mixture models was developed.  In the online setting, 
some recent contributions include \citet{wang+pb11} and \citet{bryant+s12} 
who consider online algorithms for the hierarchical Dirichlet process, and various methods inspired by the stochastic variational inference framework of 
\citet{hoffman+bwp13} (for example, \citet{wang+b12}).  \citet{tchumtchoua+d16} consider a fast online approach to fitting high-dimensional correlated data 
with a model incorporating some Bayesian nonparametric components;  their method is a variational Bayes algorithm which is similar in approach 
to methods originally developed by \citet{sato01}.  \citet{luts+bw14} consider online approaches to fitting semiparametric regression models in
the variational Bayes framework.  

The next section describes the matrix-variate Dirichlet process mixture model that is considered throughout the rest of the article.  In Section 3, a batch variational algorithm for the model is derived and then Section 4 discusses the VSUGS online algorithm which is able to work efficiently for very large datasets.  Section 5 discusses predictive inference and our novel regression adjustment approach.  Section 6 considers an application to weak informative prior selection, Section 7 considers predictive performance of the methods in some benchmark data sets and Section 8 concludes.  

\section{Matrix-variate Dirichlet process mixture model}\label{DPMMmodel}

We consider the matrix-variate Dirichlet process mixture model of \citet{Zhang2010}.  Specifically, let $y_i$, $i=1,\dots,n$ 
denote a collection of observed $m$-dimensional response vectors and $x_i$, $i=1,\dots,n$ denote corresponding $p$-dimensional
vectors of covariates.  A common flexible way to model the mean in a multivariate regression for the responses involves using some
basis expansion where, denoting the $j$th element of $y_i$ by $y_{ij}$, 
\begin{eqnarray}
 E(y_{ij}) & = & \beta_{0,j}+\sum_{r=1}^N\beta_{r,j}E_r(x_i)  \label{mvregn}
\end{eqnarray}
where $E_r(x)$, $r=1,\dots,N$ are basis functions and $\beta_j=(\beta_{0j},\dots,\beta_{Nj})^T$ are coefficients, $j=1,\dots,m$.  In motivating
their approach \citet{Zhang2010} discuss such a basis expansion, and consider setting $N=n$ and $E_r(x)=K(x,x_r)$ where $K(\cdot,\cdot)$ is a kernel
function so that the number of basis terms equals the number of observations.  Here we will be concerned with an online implementation of
their approach where $n$ is not known beforehand, so we will make a fixed choice of both $N$ and the basis functions $E_r(x)$, $r=1,\dots,N$.  We give
more details about this later.  

Write $\beta=[\beta_1,\dots,\beta_m]$ for the $(N+1)\times m$ matrix of regression coefficients and 
$E_i=(1,E_1(x_i),\dots,E_N(x_i))^T$.  Then if we assume i.i.d errors in the regression (\ref{mvregn}) we can write
$$y_i=\beta^T E_i+\epsilon_i$$
where the $\epsilon_i$ are the errors having mean $0$ and covariance matrix $\tau \Sigma$ say where $\tau>0$ is a scale parameter.  The reason for parametrizing the covariance
matrix in this way will become clear later when conjugate prior specifications are considered.  Flexible multivariate regression approaches using
basis expansions of this type have been considered by many authors.  The innovation of \citet{Zhang2010}  is to consider a model
in which the coefficient $\beta$ varies randomly between observations.  The distribution of this coefficient is estimated from the data, and 
is given a Dirichlet process prior with a matrix-variate normal distribution as the base measure.  
That is, the Dirichlet process with matrix-variate normal base measure is used as a prior on the mixing distribution for the coefficient.  
The clustering property of the Dirichlet process
ensures that many observations will share the same coefficient matrix and there is borrowing of strength both between observations and
responses in estimating the regression.  

Precisely, the model is 
\begin{align}
Y_i|E_i,\tilde{\beta}_i,\Sigma &\sim N(\tilde{\beta}_i^T E_i,\tau \Sigma)  \nonumber \\ 
\tilde{\beta}_i|Q & \sim Q \label{DPmodel} \\ 
Q|\alpha,M &\sim DP(\alpha,M) \nonumber
\end{align}
where $DP(\alpha,M)$ denotes the Dirichlet process with
precision parameter $\alpha$ and base measure $M$.  The base measure 
$M$ in the model is chosen to be a matrix-variate normal distribution $N_{N+1,m}(0,\Omega\otimes \Sigma)$. An $s\times t$ random matrix $Z$ has
a matrix-variate normal distribution $N_{s,t}(C,V\otimes W)$, where $C$ is an $s\times t$ matrix and $V$ and $W$ are $s\times s$ and $t\times t$ covariance
matrices respectively, if its density takes the form
$$p(Z)=(2\pi)^{-st/2}|V|^{-t/2}|W|^{-s/2}\exp\left(\tr\left(-\frac{1}{2}V^{-1}(Z-C)W^{-1}(Z-C)^T\right)\right).$$
In our model following \citet{Zhang2010} it will be assumed that $\Omega$ is diagonal, $\Omega=\mbox{diag}(\omega_1,\dots,\omega_{N+1})$ where $\omega_i\sim IG(a_i,b_i)$ with
$a_i$ and $b_i$ known.  Also, $\Sigma$ is inverse-Wishart with degrees of freedom $\nu$ and scale matrix $S$.  $\tau$ is given an inverse gamma prior $IG(a_\tau,b_\tau)$ with
$a_\tau$ and $b_\tau$ known.

The Dirichlet process puts all its mass on a countable collection of points so we can rewrite the model in the following way.  Let 
$\{\beta_i\}_{i=1}^\infty$ be the distinct values appearing in the sequence $\{\tilde{\beta}_i\}_{i=1}^\infty$ with the $\beta_i$ indexed according to their
order of occurrence in $\{\tilde{\beta}_i\}_{i=1}^\infty$.  We let $\delta_i$ be an integer valued variable with $\delta_i=j$ if $\tilde{\beta}_i=\beta_j$.  Write
$\delta_{1:i}=(\delta_1,\dots,\delta_{i})^T$.  Using the  P\'{o}lya urn representation for the Dirichlet process we can rewrite the model
in the form  
\begin{align}
\begin{array}{rl}
Y_i|E_i,\beta,\Sigma,\delta_i & \sim N(\beta_{\delta_i}^T E_i,\tau \Sigma)   \\
p(\delta,\beta) & =p(\delta)p(\beta)
\end{array} \label{polya}
\end{align}
where $p(\delta_{1:n})=p(\delta_{1:n}|\alpha)=\prod_{i=1}^n p(\delta_i|\delta_{1:i-1},\alpha)$, $p(\beta)=\prod_{i=1}^\infty p(\beta_i)$
with $p(\beta_i)$ the matrix-variate normal density $N_{N+1,m}(0,\Omega\otimes \Sigma)$, the priors on $\Omega$ and $\Sigma$ are the same
as before and the conditional densities $p(\delta_i|\delta_{1:i-1},\alpha)$ are defined by (using similar notation to \citet{Zhang2010})
$$p(\delta_i=j|\delta_{1:i-1},\alpha)=\left\{\begin{array}{ll}
  \frac{n_j^{(i)}}{\alpha+i-1} & \mbox{$j\in \{1,\dots,n_i\}$} \\
 \frac{\alpha}{\alpha+i-1} & \mbox{$j=n_i+1$} \end{array}\right.$$
where $n_j^{(i)}$ is the number of $\delta_k$, $k<i$ equal to $j$ and $n_i$ is the number of distinct $\beta_k$ appearing up to time $i-1$.  
For the purpose of developing our fast online variational approximation algorithm we will use a truncated Dirichlet process mixture model.  In this model
the sequence $\{\beta_i\}_{i=1}^\infty$ is truncated to $\{\beta_i\}_{i=1}^T$ where $T$ is the truncation point and
\begin{align}
p(\delta_i=j|\delta_{1:i-1},\alpha) & =\left\{\begin{array}{ll} \label{Polya}
  \frac{n_j^{(i)}+\alpha/T}{\alpha+i-1} & \mbox{$j\in \{1,\dots,n_i\}$} \\
 \frac{\alpha(1-n_i/T)}{\alpha+i-1} & \mbox{$j=n_i+1$} \end{array}\right. .
\end{align}
This is the model we discuss in what follows.

\section{Variational inference}\label{VBbatch}

Consider a Bayesian model with parameter $\xi$, prior $p(\xi)$ and  likelihood $p(y | \xi)$. Variational Bayes computational methods 
\citep{waterhouse+mr96,Jordan1999,attias00,ormerod+w10} attempt to approximate 
the posterior density $p(\xi|y)$ by a more tractable and manageable variational density $q (\xi)$, belonging to a convenient family. 
The choice of $q (\xi)$ within the approximating family is usually made by minimizing the KL divergence between $p(\xi|y)$ and $q (\xi)$. 
It can be shown that
\begin{align}\label{VB_KL}
\log p(y) &= \int \log \left( \frac{p(\xi) p(y | \xi)}{q(\xi)} \right) q(\xi) d\xi  + \int \log \left( \frac{q(\xi)}{p(\xi|y)} \right) q(\xi) d\xi
\end{align}
where $p(y)=\int p(\xi)p(y|\xi)\;d\xi$.  
The first and second terms on the RHS of \eqref{VB_KL} are the variational lower bound $\mathcal{L}$ (so-called because it forms a lower bound on $\log p(y)$) and the KL divergence between $q(\xi)$ and $q(\xi|y)$, respectively. From \eqref{VB_KL}, it is clear that minimizing the KL divergence is equivalent to maximizing $\mathcal{L}$.  For further
background see the references above.

Now, suppose that $\xi$ can be partitioned into $J$ subvectors, $\xi_1,..,\xi_J$. 
In variational Bayes, an approximating family for the posterior is considered where $q(\xi)$ is assumed to factorize as $\prod_{j=1}^J q(\xi_j)$.  
For each of the factors $q(\xi_j)$, the lower bound is maximized with the other factors held fixed by choosing $q(\xi_j)$ as 
\begin{align}
\hat{q} (\xi_j) \propto \exp \left\{ E_{-\xi_j} \log p(y|\xi) p(\xi) \right\} \label{VBmain}
\end{align}
where $E_{-\xi_j}$ denotes an expectation with respect to $\prod_{i \ne j} q (\xi_i)$. 
Expression (\ref{VBmain}) is the basis of a blockwise gradient descent algorithm for maximizing $\mathcal{L}$ where an initial choice is made for the factors and then each factor
is updated in turn with the others fixed at current values until convergence.  

One useful application of the variational approach is to approximate the posterior distribution of parameters in Bayesian nonparametric models. 
It is well known that there is usually no direct way to compute the posterior distribution in these models and that MCMC sampling methods for such models can be difficult and
computationally expensive.  These considerations motivated \citet{blei+j06} to consider a mean-field variational inference algorithm for Dirichlet process mixture models.  
Their approach can be implemented for the model of Section \ref{DPMMmodel}, since the approach of \citet{Zhang2010} is based on an ordinary Dirichlet process mixture model, and
we do implement such an approach later in our examples.  Since this is a straightforward application of the algorithm of \citet{blei+j06} we do not give further details
of their method here.  However, we develop an alternative batch variational Bayes algorithm which is also described in the next section.  
The algorithm of \citet{blei+j06} is based on the stick breaking representation of the Dirichlet process;  our alternative batch variational Bayes algorithm (like the later sequential
algorithm of Section 4) is based on the P\'{o}lya urn representation with the unknown mixing distribution integrated out.  Although the alternative batch algorithm
involves some further approximations, the purpose of developing this method is that it gives a batch algorithm similar to our later online approach, and 
provides another reference for comparison for the performance of the online algorithm where how much performance is lost through the sequential updating mechanism can
be better understood.  Also, many of the updating steps in the online algorithm are simple modifications of the corresponding steps for the batch algorithm.


\subsection{Batch mean field updates for global parameters}

We work with the model (\ref{polya}).
For the matrix-variate DP mixture model, variational inferences for the parameters $\beta_{1:T}, \Sigma, \tau, \omega_{1:N}, \delta_{1:n}$ are required. Define $\theta = ( \beta_{1:T}, \Sigma, \tau, \omega_{1:N} )^T$ and $\theta_j =  ( \beta_{j}, \Sigma, \tau, \omega_{1:N} )^T$ for $j=1,..,T$. In deriving approximate mean field updates we consider a slight expansion of the model (\ref{polya}). This will be helpful when discussing the online case later, since in our model expansion the variational posterior has the same form as the prior leading to a natural online implementation.  
In the expanded model the prior on $\beta_i$ is changed from $N_{N+1,m} (0,\Omega\otimes \Sigma)$ to $N_{N+1,m}(M_i,\Omega_i\otimes \Sigma)$ with $\Omega_i=(\Omega^{-1}+C_i)^{-1}$ where $M_i$ and $C_i$ are known matrices. Letting $\beta=(\beta_1,\dots,\beta_T)$, we consider the following factorization for the variational posterior distribution:
$$q(\theta,\delta_{1:n})=q(\beta,\Sigma)q(\tau) q(\omega_{1:N+1})q(\delta_{1:n}).$$
$q(\delta_i=j)$ will be denoted by $q_{ij}$.   In this subsection we give the mean field updates for all factors except for $q(\delta_{1:n})$, which is considered
in the next subsection.  Technical details of the derivations are found in Appendix A.

For $\beta$, we recognize the form of $q(\beta,\Sigma)$ as being $q(\beta,\Sigma)=q(\Sigma)q(\beta|\Sigma)$ where $q(\Sigma)$ is inverse Wishart, and $q(\beta|\Sigma)=\prod_{j=1}^T q(\beta_j|\Sigma)$ with $q(\beta_j|\Sigma)=N_{N+1,m}(\hat{\beta}_j,V_j^{-1}\otimes \Sigma)$,
$$ \hat{\beta}_j=V_j^{-1}\left((E_q(\Omega^{-1})+C_j)M_j+E_q(\tau^{-1})\sum_{i=1}^n q_{ij}E_iy_i^T\right) $$
and
$$V_j=\left(E_q(\Omega^{-1})+C_j+E_q(\tau^{-1})\sum_{i=1}^n q_{ij}E_iE_i^T\right). $$
For $\Sigma$, $q(\Sigma) = IW (\hat{\nu}, \hat{S})$ where 
$$\hat{\nu} = \nu+n, \hat{S} = S+\sum_{j=1}^T \left(E_q(\tau^{-1})\sum_{i=1}^n q_{ij}y_iy_i^T+M_j^T
(E_q(\Omega^{-1})+C_j)M_j-\hat{\beta}_j^T V_j \hat{\beta}_j\right).$$
For $\tau$, $q(\tau) = IG(\hat{a}_{\tau}, \hat{b}_{\tau})$ where
$$ \hat{a}_\tau = a_\tau+\frac{nm}{2}, \hat{b}_\tau = b_\tau+\frac{1}{2}\sum_{i=1}^n\sum_{j=1}^T q_{ij}((y_i-\hat{\beta}_j^TE_i)^TE_q(\Sigma^{-1})(y_i-\hat{\beta}_j^TE_i)+mE_i^TV_j^{-1}E_i). $$
Lastly, for $\omega = (\omega_1,..,\omega_{N+1})$, $q(\omega)=\prod_{i=1}^{N+1}q(\omega_i)$. Each $q(\omega_i) = IG(\hat{a}_i, \hat{b}_i)$ where
\begin{align}\label{batchomega}
\hat{a}_i = a_i+mT/2, \hat{b}_i = b_i+\frac{1}{2}\sum_{j=1}^T \left((\hat{\beta}_{j,i}-M_{j,i}) E_q(\Sigma^{-1})(\hat{\beta}_{j,i}-M_{j,i})^T+m \omega_{ij}'\right). 
\end{align}
where $\hat{\beta}_{j,i}$ is the $i$th row of $\hat{\beta}_j$ and $\omega_{ij}'$ is the $i$th diagonal element of $V_j^{-1}$.

\subsection{Batch mean field update for local parameters}\label{Batch_deltaupdate}

We now factorize $q(\delta_{1:n})$ as $\prod_{i=1}^n q(\delta_i)$ and consider approximate mean field updates for $q(\delta_i)$, $i=1,\dots,n$.  
Using \eqref{VBmain}, for each $\delta_i$, we get
\begin{align*}
q (\delta_{i}) &\propto \exp\left(E_q \left\{ \log p(\delta_{i} | \delta_{\neq i}, \alpha) \right\} + E_q \left\{ \log (y_{i} | \theta, \delta_{i}) \right\}\right).
\end{align*}
where $\delta_{\neq i}$ denotes $\delta_{1:n}$ with $\delta_i$ omitted.  
Making the approximation 
\begin{align*}
 \exp\left(E_q\left\{\log p(\delta_i|\delta_{\neq i},\alpha)\right\}\right) & \approx E_q\left\{p(\delta_i|\delta_{\neq i},\alpha)\right\}
\end{align*}
we have 
\begin{align*}
q (\delta_{i}) &\propto E_q\left\{p(\delta_i|\delta_{\neq i},\alpha)\right\} \exp\left(E_q \left\{ \log (y_{i} | \theta, \delta_{i}) \right\}\right).
\end{align*}
If $i\geq T$, we approximate further
$E_q\left\{p(\delta_i=j|\delta_{\neq i},\alpha)\right\}$ by 
$$\frac{\sum_{k\neq i} q_{kj}+\alpha/T}{\alpha+n-1},$$
with the case where $i<T$ being handled by using the same expression but conditioning on $\delta_i\leq i$.  
This approximation is obtained by reordering so that the $i$th observation is last, taking an expectation in (\ref{Polya}) and then restoring the constraint associated with
the original ordering by conditioning on $\delta_i\leq i$ if $i\leq T$.  Note that because we order atoms according to their order of occurrence it must be the case that $\delta_i\leq i$ for
$i\leq T$.  To get an expression for our approximate mean field update it remains to evaluate 
$E_q \left\{ \log(y_i | \theta_{j} ) \right\}$ which is 
\begin{align*}
E_q \left\{ \log(y_i | \theta_{j} ) \right\} &= -\frac{m}{2} \log (2 \pi) - \frac{m}{2} E_q \left\{ \log \tau \right\} - \frac{1}{2} E_q \left\{ \log |\Sigma| \right\} \\
&\qquad -\frac{1}{2} E_q \left\{ \frac{1}{\tau} \right\} E_q \left\{ y_i^T\Sigma^{-1} y_i - 2 E_i^T \beta_j \Sigma^{-1} y_i + E_i^T \beta_j \Sigma^{-1} \beta_j^T E_i \right\} \\ 
&= -\frac{m}{2} \log (2 \pi) - \frac{m}{2} \left\{  \log (\btauhat) - \psi(\atauhat) \right\} \\
&\qquad - \frac{1}{2} \left\{ - \psi_m \left(\frac{\nuhat}{2} \right) - m \log (2) + \log |\Shat| \right\}  -\frac{1}{2} \frac{\atauhat}{\btauhat}  \left\{ \nuhat y_i^T (\Shat)^{-1} y_i  \right. \\
&\qquad \left. - 2 \nuhat  E_i^T \hat{\beta}_j (\Shat)^{-1} y_i + m E_i^T \Omegajhat E_i + (E_i^T \bhat_j)^T (\nuhat (\Shat)^{-1}) (E_i^T \hat{\beta}_j) \right\}. 
\end{align*}

\section{VSUGS for matrix-variate Dirichlet process mixture model}\label{matrixVSUGS}

The VSUGS algorithm, proposed by \citet{zhang+nyj14}, is an online learning procedure for fast fitting of Dirichlet process mixture models. It uses the variational approximation framework to improve the SUGS algorithm \citep{wang+d11}. The VSUGS algorithm is especially useful for large datasets as computing the full variational batch update or using MCMC 
might be computationally infeasible. The framework of the VSUGS procedure is as follows. Following \citet{zhang+nyj14}, we consider an approximation to the posterior $p(\delta_{1:i-1}, \theta_{1:T} |y_{i:i-1})$ of the form
\[
\prod_{j=1}^{i-1} q_{i-1} (\delta_j) \prod_{j=1}^{T} q_{i-1} (\theta_j).
\]
The algorithm starts at $\hat{q}_1 (\delta_1 = 1) = 1$, $\hat{q} (\theta_1) = p(\theta_1 | y_1, \delta_1 = 1)$. Then, at time $i$, we use $\hat{q}_{i-1} (\theta)$ and $\hat{q}_{i-1} (\delta_{1:i-1})$ as a prior for processing the data point $y_i$.  Then for a certain fixed choice of $\hat{q}_i(\delta_i)$ the mean field update for $\theta$ reduces to the following
approximation of $p(\theta|y)$:
\begin{align}\label{onlineupdate}
\hat{q}_i (\theta) \propto \hat{q}_{i-1}(\theta) \prod_{j=1}^T \exp\left( \hat{q}_i(\delta_i=j) E_q\left\{\log p(y_i|\theta_j)\right\}\right).
\end{align}
For the assignment variables $\delta_i$, we follow \citet{zhang+nyj14} and choose
\begin{align}
\hat{q}_i (\delta_i = j) = r_{ij} \int \hat{q}_{i-1} (\theta_{\delta_i}) p(y_i | \theta_{\delta_i}) d \theta_{\delta_i} \label{VSUGSdelta}
\end{align}
for $ j \in \{ 1,..., \min(i,T)\}$ where $T$ is the pre-specified truncation point for the number of mixture components and
\begin{align}
r_{ij}=\left\{\begin{array}{ll}
  \frac{ \sum_{k=1}^{i-1} \hat{q}_{i-1} (\delta_k = j) +\alpha/T}{\alpha+i-1} & \mbox{$j\in \{1,\dots, \min\left(i-1,T \right) \}$,} \\
 \frac{\alpha(1- \min\left(i-1,T \right) /T)}{\alpha+i-1} & \mbox{$j=\min\left(i-1,T \right) +1$.} \end{array}\right. \label{qij}
 \end{align}
One property of the VSUGS procedure is that \eqref{VSUGSdelta} splits the likelihood contribution from the $i$th observation among the mixture components. This deviates from the original SUGS algorithm \citep{wang+d11} which uses a ``hard" allocation to mixture components.  In the case of conjugate priors, the VSUGS algorithm retains the computational advantages
of the original SUGS algorithm.  See \citet{zhang+nyj14} for further details.

\subsection{Sequential update of variational parameters for $\tau, \Sigma$ and $\Omega$}

In the batch update of the global parameters the expectations of $\tau^{-1}$,  $\Sigma^{-1}$ and $\Omega^{-1}$ with respect to $q (\tau)$, $q (\Sigma)$ and $q (\Omega)$ respectively are required. For an online algorithm like VSUGS, these expectations change when a new data point enters. In order to use \eqref{onlineupdate}, 
it is required to replace the expectation of $\tau^{-1}, \Sigma^{-1}$ and $\Omega^{-1}$ with $E_{q_i} (\tau^{-1}), E_{q_i} (\Sigma^{-1})$ and $E_{q_i} (\Omega^{-1})$ respectively, where $E_{q_i}$ represents the variational expectation at time $i$. Following the derivation of the batch updates, our corresponding online learning update for the variational parameters is $q_i (\beta_j) \sim N_{q+1,m} (\bhatji, (V_j^{(i)})^{-1} \otimes \Sigma)$,  $q_i (\Sigma) \sim IW\left(\nui ,S^{(i)} \right)$. $q_i (\tau) \sim IG(\ataui, \btaui)$, where

\begin{align*}
\hat{\beta}^{(i)}_j &= (V_j^{(i)})^{-1}\left(   (V_j^{(i-1)}) \bhatjione + \frac{\atauione}{\btauione} \hat{q}_i (\delta_i = j) E_iy_i^T\right), \\
V_j^{(i)} &= V_j^{(i-1)} + \frac{\atauione}{\btauione}\hat{q}_i (\delta_i = j) E_iE_i^T , 
\end{align*}
\begin{align*}
 \nui &= \nu + i.  \\
\Si &=  \Sione +    \sum_{j=1}^T \left\{ \frac{\atauione}{\btauione} \hat{q}_i (\delta_i = j) y_iy_i^T + (\bhatjione)^T  V^{(i-1)}_j \bhatjione -(\bhatji)^T V^{(i)}_j \bhatji \right\}, \\
\ataui &=  a_{\tau} +\frac{im}{2}, \\   
\btaui  &= \btauione  +  \frac{1}{2} \sum_{j=1}^{T} \hat{q}_i (\delta_i = j) \left\{ (y_i- (\bhatji)^TE_i)^T  \nui (\Si)^{-1}  (y_i  - (\bhatji)^TE_i) \right. \\
&\qquad\qquad\qquad\qquad\qquad\qquad\qquad \left.+mE_i^T(V^{(i)}_j)^{-1}E_i) \right\}. 
\end{align*}
As the second term on the RHS of \eqref{onlineupdate} does not include any terms for $\omega_1,...,\omega_{N+1}$, there is no online learning required for these parameters. At every step of the online VSUGS algorithm, we continue using the batch update for $q_i (\omega_{1:N+1})$. For the full algorithm, we refer to Algorithm \ref{MVSUGS}.

\subsection{Sequential VSUGS type update for the $\delta_i$}

Suppose we assimilate observations sequentially and at step $i-1$ we have a variational posterior distribution of the form
$$q_{i-1}(\beta,\Sigma)q_{i-1}(\tau)q_{i-1}(\omega)q_{i-1}(\delta)$$
where $q_{i-1}(\beta,\Sigma)=q_{i-1}(\Sigma)\prod_{j=1}^T q(\beta_j|\Sigma)$ with $q_{i-1}(\Sigma)$ being $IW(\nuione,\Sione)$, 
$q_{i-1}(\beta_j|\Sigma)$ being $N_{N+1,m}(\bhatjione,\Omegajione \otimes \Sigma)$, 
$q_{i-1} (\tau)=IG(\atauione,\btauione)$ and $q_{i-1}(\omega)=\prod_{j=1}^{N+1}q_{i-1}(\omega_j)$ with $q_{i-1} (\omega_j)$ being $IG(\ajione,\bjione)$.  
Also $q_{i-1}(\delta)=\prod_{j=1}^{i-1}q(\delta_j)$. Using the VSUGS approximation, we take
\begin{align}\label{qij}
\hat{q}_i (\delta_i = j) = r_{ij} \int p(y_i|\beta_{\delta_i}^T E_i,\tau\Sigma)q_{i-1}(\tau)q_{i-1}(\beta_{\delta_i}|\Sigma)q_{i-1}(\Sigma)d\beta_{\delta_i} d\Sigma d\tau
\end{align}
where
\begin{align}
r_{ij}=\left\{\begin{array}{ll}
  \frac{ \sum_{k=1}^{i-1} \hat{q}_{i-1} (\delta_k = j) +\alpha/T}{\alpha+i-1} & \mbox{$j\in \{1,\dots, \min\left(i-1,T \right) \}$} \\
 \frac{\alpha(1- \min\left(i-1,T \right) /T)}{\alpha+i-1} & \mbox{$j=\min\left(i-1,T \right) +1$} \end{array}\right. . \label{qij}
 \end{align}
%
The integral in \eqref{qij} can be evaluated as (see Appendix B)
\begin{align}
& \int p(y_i|\beta_j^TE_i,\tau\Sigma)q_{i-1}(\tau)q_{i-1}(\beta_j|\Sigma)q_{i-1}(\Sigma)d\beta_jd\Sigma d\tau   \nonumber \\
& = \int (2\pi \tau)^{-\frac{m}{2}}q_{i-1}(\tau)|\Omegajione|^{-\frac{m}{2}}|\frac{1}{\tau}E_iE_i^T+(\Omegajione)^{-1}|^{-\frac{m}{2}} \label{inttau} \\
& \qquad \frac{|\Sione|^{\nuione/2}\Gamma_m((\nuione+1)/2)}{2^{\nuione \frac{m}{2}}\Gamma_m(\nuione/2)} 2^{(\nuione+1)\frac{m}{2}}  \nonumber  \\
&  \qquad \left |\Sione +\frac{1}{\tau}y_iy_i^T+ (\bhatjione)^T (\Omegajione)^{-1}\bhatjione -\bar{\beta}_j^T(\frac{1}{\tau}E_iE_i^T+ (\Omegajione)^{-1})\bar{\beta}_j \right|^{-\frac{\nuione+1}{2}}d\tau \nonumber
\end{align}
where $\bar{\beta}_j= \left(\frac{1}{\tau}E_iE_i^T+ (\Omegajione)^{-1}\right)^{-1} \left(\frac{1}{\tau}E_iy_i^T+  (\Omegajione)^{-1} \bhatjione \right)$. This last integral does not seem to be easily computable analytically.  It is an expectation with respect to $q_{i-1}(\tau)$, and if this distribution is concentrated around the mean 
it is reasonable to make the approximation $\int f(\tau)q_{i-1}(\tau)d\tau=f(E_{q,i-1}(\tau))$ for functions $f(\tau)$ and where we have written $E_{q,i-1}(\tau)=\int \tau q_{i-1}(\tau)d\tau$. Using this approximation here we get that the integral is approximately
\begin{align}
&  \left(\frac{2\pi}{\mutauione} \right)^{-m/2}|\Omegajione|^{-m/2}| \mutauione E_iE_i^T+ (\Omegajione)^{-1}|^{-m/2}  \label{qdelta} \\
& \frac{|\Sione|^{\nuione /2} \Gamma_m\left(\frac{\nuione+1}{2}\right)}{2^{\nuione m/2}\Gamma_m(\nuione /2)} 2^{\frac{(\nuione+1)m}{2}} \nonumber \\ 
& \left|\Sione+ \mutaui y_iy_i^T+ (\bhatjione)^T (\Omegajione)^{-1}\bhatjione -\tilde{\beta}_j^T(\mutaui E_iE_i^T+(\Omegajione)^{-1})\tilde{\beta}_j \right|^{-\frac{\nuione +1}{2}} \nonumber
\end{align}
where $\tilde{\beta}_j= \left( \mutauione E_iE_i^T+ (\Omegajione)^{-1}\right)^{-1} \left( \mutauione E_iy_i^T+  (\Omegajione)^{-1} \bhatjione \right)$ and $\mutauione$ is the expectation of $\tau^{-1}$ with respect to $q_{i-1} (\tau)$.

\begin{algorithm}
\floatname{algorithm}{Algorithm}
 \caption{\bf : VSUGS algorithm for Matrix DPMM}\label{MVSUGS}
\begin{algorithmic}
\State Initialize $\alpha, T, a_i,  b_i , a_{\tau}, b_{\tau}, \nu, S, C_1,..,C_T, M_1,..,M_T$
\State $\hat{q}_1 (\delta_1 = 1) \gets 1, \hat{q}_1 (\delta_1 = 2), \dots, \hat{q}_1 (\delta_1 = T) \gets 0$. 
\State $a^{(0)}_{\tau} \gets a_{\tau} $, $b^{(0)}_{\tau} \gets b_{\tau}$, $a^{(0)}_i \gets a_i$, $b^{(0)}_i \gets b_i $, $S^{(0)} \gets S$, $\nu^{(0)} \gets \nu$.
\State $\mu^{(0)}_{\tau^{-1}} \gets \frac{a_{\tau}^{(0)} }{b_{\tau}^{(0)}}$, $\mu^{(0)}_{\Sigma^{-1}}  \gets \nu^{(0)} (S^{(0)})^{-1}$, $\mu_{\Omega^{-1}}^{(0)} \gets \mbox{diag}(a^{(0)}_i / b^{(0)}_i)$.
\State $V^{(0)}_j \gets\mu_{\Omega^{-1}}^{(0)}+ C_j$, $\hat{\beta}_j^{(0)} \gets M_j$.
\For{ $i = 1:n$}
\If{$i \geq 2$}
\State $T_i \gets \min(T, i-1)$.
\For{ $j = 1:T_i$}
\If{$j < T_i$}
\State $r_{ij} \gets   \frac{ \sum_{k=1}^{i-1} \hat{q}_{i-1} (\delta_k = j) +\alpha/T}{\alpha+i-1}  $.
\Else
\State $r_{ij} \gets  \frac{\alpha(1- \min\left(i-1,T \right) /T)}{\alpha+i-1} $.
\EndIf
\State $\Lambda \gets \mutauione E_iE_i^T+(\Omegajione)^{-1}$.
\State  $\tilde{\beta}_j \gets \Lambda \left(\mutauione E_iy_i^T+  (\Omegajione)^{-1} \bhatjione \right) $.
\State $\hat{q}_i (\delta_i = j) \gets r_{ij} (\pi (\mutauione)^{-1})^{-m/2} |\Omegajione|^{-m/2} |\Lambda|^{-m/2} \frac{|\Sione|^{\nuione/2}}{\Gamma_m(\nuione/2)} 
\Gamma_m\left(\frac{\nuione+1}{2}\right)$
\State \qquad\qquad $\left|\Sione + \mutauione y_iy_i^T+ (\bhatjione)^T(\Omegajione)^{-1}\bhatjione -\tilde{\beta}_j^T \Lambda \tilde{\beta}_j \right|^{-\frac{\nuione+1}{2}}. $
\EndFor
\EndIf
\For{ $j = 1:T$}
\State $V^{(i)}_j \gets V^{(i-1)}_j + \mutauione \hat{q}_i (\delta_i = j) E_iE_i^T$.
\State $\bhatji \gets (V^{(i)}_j)^{-1} \left(  V^{(i-1)}_j \bhatjione + \mutauione \hat{q}_i (\delta_i = j) E_i y_i^T  \right)$.
\EndFor
\State $\nui \gets \nuione + 1$.
\State $\Si \gets \Sione +    \sum_{j=1}^{T} \left\{  \mutauione \hat{q}_i (\delta_i = j) y_iy_i^T + (\bhatjione)^T  V^{(i-1)}_j \bhatjione -(\bhatji)^T V^{(i)}_j \bhatji \right\} $.
\State $\ataui \gets \atauione + \frac{m}{2}$
\State $\btaui \gets \btauione $
\State \qquad  $ + \frac{1}{2} \sum_{j=1}^{T} \hat{q}_i (\delta_i = j) \left\{ (y_i- (\bhatji)^TE_i)^T \musigmai  (y_i- (\bhatji)^TE_i)+mE_i^T(V^{(i)}_j)^{-1}E_i) \right\}.$
\State $\mutaui \gets \frac{\ataui}{\btaui}$, $\musigmai \gets \nui (\Si)^{-1}$.
\For{ $k = 1:N+1$}
\State $a_{k}^{(i)} \gets a_{k}^{(0)} + \frac{m T}{2}$
\State $b_k^{(i)} \gets  b_{k}^{(0)} +\frac{1}{2}\sum_{j=1}^{T} \left((\hat{\beta}_{j,k}^{(i)}-M_{j,k}^{(i)})^T  \musigmai  (\hat{\beta}_{j,k}^{(i)}-M_{j,k}^{(i)})+m \omega_{ij}'\right)$, 
\State \qquad where $\omega_{ij}$ is the $i$th diagonal element of $V^{(i)}_j$  and $\hat{\beta}_{j,k}$ is the $k$th row of $\hat{\beta}_j$. 
\EndFor
\EndFor
\end{algorithmic}
\end{algorithm}

\section{Posterior predictive inference}\label{Post_pred}

Suppose we are given a new input vector $x_0$ and wish to predict the response vector $y_0$. 
Write $E_0=(1,E_1(x_0),\dots,E_N(x_0))^T$.  
The posterior predictive distribution of $y_0$ can be evaluated as
\begin{align*}
p(y_0 | x_0, y_{1:n}) = \sum_{j=1}^T p(\delta_0 = j | y_{1:n}) \int p(y_0 | y_{1:n}, \tau, \Sigma, \beta_{1:T}, \delta_0 = j)p( \tau, \Sigma, \beta_{1:T}| y_{1:n}) d \beta_{1:T} d\tau d\Sigma.
\end{align*}
Assuming $n>T$, we replace $p(\delta_0 = j | y_{1:n})$ with $r_{n+1\,j}$ as defined in \eqref{qij}. Also, replacing $p(\tau, \beta_j, \Sigma | y_{1:n})$ with the 
corresponding variational posterior $q_n(\tau, \beta_j, \Sigma)$, the predictive density becomes
\begin{align}
p( y_0 | E_0, y_{1:n}) = \sum_{j=1}^T r_{n+1\,j} \int p( y_0 | y_{1:n}, \tau, \Sigma, \beta_j, \delta_0 = j) q_n(\tau ) q_n(\beta_j) q_n(\Sigma)    d\beta_{j} d\tau d\Sigma. \label{predictive}
\end{align}
The integral in \eqref{predictive} evaluates to a multivariate $t$-distribution.  So an approximate posterior predictive density is obtained as a mixture
of multivariate $t$-densities.  For more details, including the parameters of the multivariate-$t$ mixture components, see Appendix C.

\subsection{Regression-type adjustment for improving predictive inference}\label{RegVSUGS}

One advantage of using the matrix-variate Dirichlet process approach to flexible regression is that avoiding covariate dependence in the mixing weights
greatly simplifies computation, something that we have exploited here for implementing an online algorithm.  However, this does place a greater burden on the
mean functions in the regression mixture components to model the response distribution in a flexible way.  Here we consider a method for improving predictive
performance of the fitted model, borrowing an idea from the literature on regression adjustment methods for approximate Bayesian computation 
\citep{beaumont+zb02,blum10,blum+f10,blum+t10}.  The idea below is given in equation (4.1) of \citet{blum+t10}.

Suppose we wish to consider prediction of a new response $y^*=(y^*_1,\dots,y^*_m)^T$ to be observed with corresponding covariate $x^*=(x^*_1,\dots,x^*_p)^T$.  
Write $N_k(x^*)$ for the $k$ nearest neighbours of $x^*$ among the observed covariates $\{x_1,\dots,x_n\}$.  We write the corresponding values of $(x,y)$ as
$(x_{i_1},y_{i_1}),\dots,(x_{i_k},y_{i_k})$ so that $i_1,\dots,i_k$ denote the indices of the covariates in $N_k(x^*)$.  In our fitted regression model, write
$\hat{F}_j(.|x)$ for the marginal distribution function of the $j$th component of the response in the fitted model at $x$, and $\hat{F}_j^{-1}(\cdot|x)$ for its inverse
where it is assumed this exists.  

If the fitted model is correct, $\hat{F}_j(y_{i_r j}|x_{ir})$ is uniform on $[0,1]$, and 
 $\hat{F}_j^{-1}(\hat{F}_j(y_{i_r j}|x_{i_r})|x^*)$ has the distribution
$\hat{F}_j(\cdot|x^*)$.  So if we set 
\begin{align}\label{RegadjVSUGS}
 y_r^a & = (y_{r1}^a,\dots,y_{rm}^a)^T \\
 & = (\hat{F}_1^{-1}(\hat{F}_1(y_{i_r 1}|x_{i_r})|x^*),\dots,\hat{F}_m^{-1}(\hat{F}_m(y_{i_r m}|x_{i_r})|x^*))^T \nonumber
\end{align}
$r=1,\dots,k$, then marginally $y_{1j}^a,\dots,y_{kj}^a$ is a sample from $\hat{F}_j(\cdot|x^*)$ (if the regression model is correct).  

The sample $y_r^a$, $r=1,\dots,k$ can be used to do approximate predictive inference.  The advantage of this method compared to using $\hat{F}_j(\cdot|x^*)$ 
directly is that by using, in effect, quantile residuals locally around $x^*$ to define the particles $y^a_r$ we are able to adjust for any local misfit of the regression
model.  This can result in improved predictive inference.  Note that by transforming the particles $y_{i_r}$ component-wise we are not guaranteed to preserve the correct
multivariate dependence structure in the fitted model at $x^*$, but if the copula of the fitted distribution changes only slowly with $x$ over the neighbourhood used 
the effects of this approximation are minor.

\section{Application to weak informative prior selection}

As an application of our proposed methodology, we consider flexible approximation of prior predictive densities as a function of a prior hyperparameter value based
on data simulated under a model, when these prior predictive densities are not analytically tractable.  Approximating such predictive densities is useful for prior choice.
In the application considered here we make use of the way that the MDP mixture model is able to approximate the whole response distribution flexibly.  In the next
section we will look more closely at the quality of point predictions of the online algorithm compared to those obtained by batch VB and MCMC approaches.  

Consider a statistical model $p(y|\xi)$ for data $y$ with parameter $\xi$.  Suppose we have a class of priors $p(\xi|\lambda)$ where $\lambda$ is a hyperparameter
value to be chosen.  We also suppose that there is a value $\lambda_0$ for $\lambda$ that has already been chosen tentatively as representing our best current
prior knowledge of $\xi$.  For the purpose of sensitivity analysis, we may wish to define a prior that is less informative than $p(\xi|\lambda_0)$, 
and this might be particularly useful in the case where the information brought by the prior and likelihood seem to be contradictory.  
\citet{Evans2011} considered defining the amount of information in a prior $p(\xi|\lambda)$ relative to $p(\xi|\lambda_0)$ through the idea
of prior-data conflict.  The notion of weakly informative priors formalized in \citet{Evans2011} was inspired by previous work of \citet{Gelman2006}.

Since the idea of \citet{Evans2011} is built on the idea of checking for prior-data conflict, this needs to be understood first.  
Prior-data conflict occurs where the prior puts all its mass out in the tails of the likelihood.  A way of testing for prior-data conflict which
modifies a suggestion of \citet{Box1980} will be considered here, following \citet{Evans2006}.  Their idea is that a minimal sufficient statistic value $S$ determines the likelihood, 
so we can check if the observed likelihood is in conflict with the prior by seeing whether the observed value of the sufficient statistic $S_{obs}$ say lies out in the tails
of its prior predictive distribution.  A $p$-value for checking for conflict with the prior $p(\xi|\lambda)$ can be computed as
\begin{align}
p(S_{obs},\lambda) & =P(p(S|\lambda)\leq p(S_{obs}|\lambda)),  \label{conflictp}
\end{align}
where $S\sim p(S|\lambda)$ and $p(S|\lambda)=\int p(S|\xi)p(\xi|\lambda) d\xi$ is the prior predictive distribution of $S$.  
If a non-trivial sufficient statistic does not exist it may be reasonable to choose an asymptotically sufficient statistic such as the maximum likelihood estimator
or some approximation to it.  Note that $p(S_{obs},\lambda)$ is calculating the probability that a random draw from $p(S|\lambda)$ has lower density than the value 
of $S_{obs}$ and it is small if $S_{obs}$ lies out in the tails of $p(S|\lambda)$.  The above prior-data conflict check can be 
modified in various ways - for more details see \citet{Evans2006}.  

To use this notion of prior-data conflict checking to define how informative the prior $p(\xi|\lambda)$ is relative to $p(\xi|\lambda_0)$ \citet{Evans2011} consider
$S$ generated randomly under $p(S|\lambda_0)$ and ask whether for data generated in such a way does doing the analysis under $p(\xi|\lambda)$ rather than $p(\xi|\lambda_0)$ result
in a reduction of the frequency of prior-data conflicts.  The occurrence of a conflict is defined by choice of a certain cutoff for a conflict $p$-value such as (\ref{conflictp}).  
It is possible to consider various modifications of the basic idea considering uniformity of reduction of levels of conflict over different $p$-value cutoffs, see
\citet{Evans2011} for more details.  

Following the ideas of \citet{Evans2006} and \citet{Evans2011}, \citet{Nott2015} propose modifying a regression adjustment approach used in the approximate Bayesian computation (ABC) literature to approximate prior predictive distributions $p(S|\lambda)$ for many different $\lambda$ in a computationally thrifty way when $S$ may be expensive to compute.  In particular, they consider the method of \citet{blum+f10}, which modifies a suggestion of \citet{beaumont+zb02}, to generate approximate samples from prior predictive densities $p(S|\lambda)$ and then use these samples for the required computations. 
This approach is much more computationally efficient than generating a large number of replications of $S$ at each value of $\lambda$ independently for every value $\lambda$ of interest
on a grid, say.  The method starts by 
generating values $\lambda_i$, $i = 1,...,n$ from a pseudo-prior $p(\lambda)$.  
Then values $(\xi_i, S_i)$, $i = 1,.., n$ are generated for$(\xi,S)$ from $ p(\xi|\lambda) p(S | \xi)$ where $S$ is a minimal sufficient statistic or some asymptotically sufficient statistic.  
\citet{Nott2015} modify the ABC with regression adjustment method in \citet{blum+f10} by reversing the usual role of the parameters and the summary statistics 
where these methods are used in the ABC context.  
They fit a regression model with
\begin{align}\label{ABCmodel}
S_i = \mu(\lambda_i) + \sigma (\lambda_i) \epsilon_i
\end{align}
where the $\epsilon_i$ are i.i.d errors with zero mean and variance one and $\mu (\lambda)$ and $\sigma (\lambda)$ are flexible mean and standard deviation functions.  
\citet{blum+f10} parametrize $\mu(\cdot)$ and $\sigma(\cdot)$ using neural networks.  
After fitting the model to the data to obtain estimates $\hat{\mu} (\lambda)$ and $\hat{\sigma} (\lambda)$, a sample of $p(S|\lambda)$ can be obtained approximately
by considering the fitted mean for the regression model plus the empirical residuals.  The empirical residual for the $i$th point is 
$\hat{\epsilon}_i=\hat{\sigma}  (\lambda_i)^{-1} (S_i - \hat{\mu} (\lambda_i) )$, and using such empirical residuals together with the fitted model at $\lambda$ gives
\begin{align}\label{ABCfit}
S_i^a (\lambda) = \hat{\mu} (\lambda) + \hat{\sigma} (\lambda)  \left\{\hat{\sigma}  (\lambda_i)^{-1} (S_i - \hat{\mu} (\lambda_i) )\right\}
\end{align}
as an approximate sample from $p(S|\lambda)$ if the regression model is correct.  
Based on the approximate sample $S_i^a(\lambda)$, they use a kernel estimate to approximate $p(S|\lambda)$. Let this kernel estimate be $\hat{p} (S|\lambda)$. 
Next, suppose that $S_j^0$, $j=1,\dots,n$ are draws from $p(S|\lambda_0)$.  Then a particle approximation to the distribution of $p(S,\lambda)$ for $S\sim p(S|\lambda_0)$ is
given by the values $\hat{P} (S_1^0,\lambda),..., \hat{P} (S_n^0,\lambda)$ where
\[
\hat{P} (S_j^0,\lambda) = \frac{1}{n} \sum_{i=1}^n I (\hat{p} (S_i^a (\lambda) | \lambda) \leq \hat{p} (S_j^0 | \lambda))
\]
The distribution of the $p$-value can be used to determine whether $p(\xi|\lambda)$ is weakly informative relative to $p(\xi|\lambda_0)$ or not.   

We propose using our approach to assess weak informativity of alternative priors compared to a base prior, similar to the above.  
However, instead of using 
the ABC with regression adjustment \eqref{ABCmodel}, we propose fitting a matrix-variate Dirichlet process mixture model with
$S$ as response and $\lambda$ as predictors. 
In applying the MDP prior approach we also employ the regression adjustment method of Section \ref{RegVSUGS} to obtain approximate samples from $p(S|\lambda)$ 
at any desired value of $\lambda$.  Kernel estimates of $p(S|\lambda)$ are then constructed as for the approach of \citet{Nott2015} and the procedure above followed
for approximating the distribution of conflict $p$-values for $S$ generated from $p(S|\lambda_0)$.  As observed in \citet{Nott2015} high accuracy is not needed in the 
regression calculations;  
the regression calculations are simply a screening computation, and 
once a candidate value of $\lambda$ is chosen for a weakly informative prior then for the single finally chosen value we can
generate a large sample from the prior predictive distribution and see whether our approximate calculations were good enough.

\subsection{Analysis of a logistic regression example}\label{bioassay}

\begin{figure}[h]
  \centering
\begin{subfigure}{.45\textwidth}
      \includegraphics[width=1\textwidth]{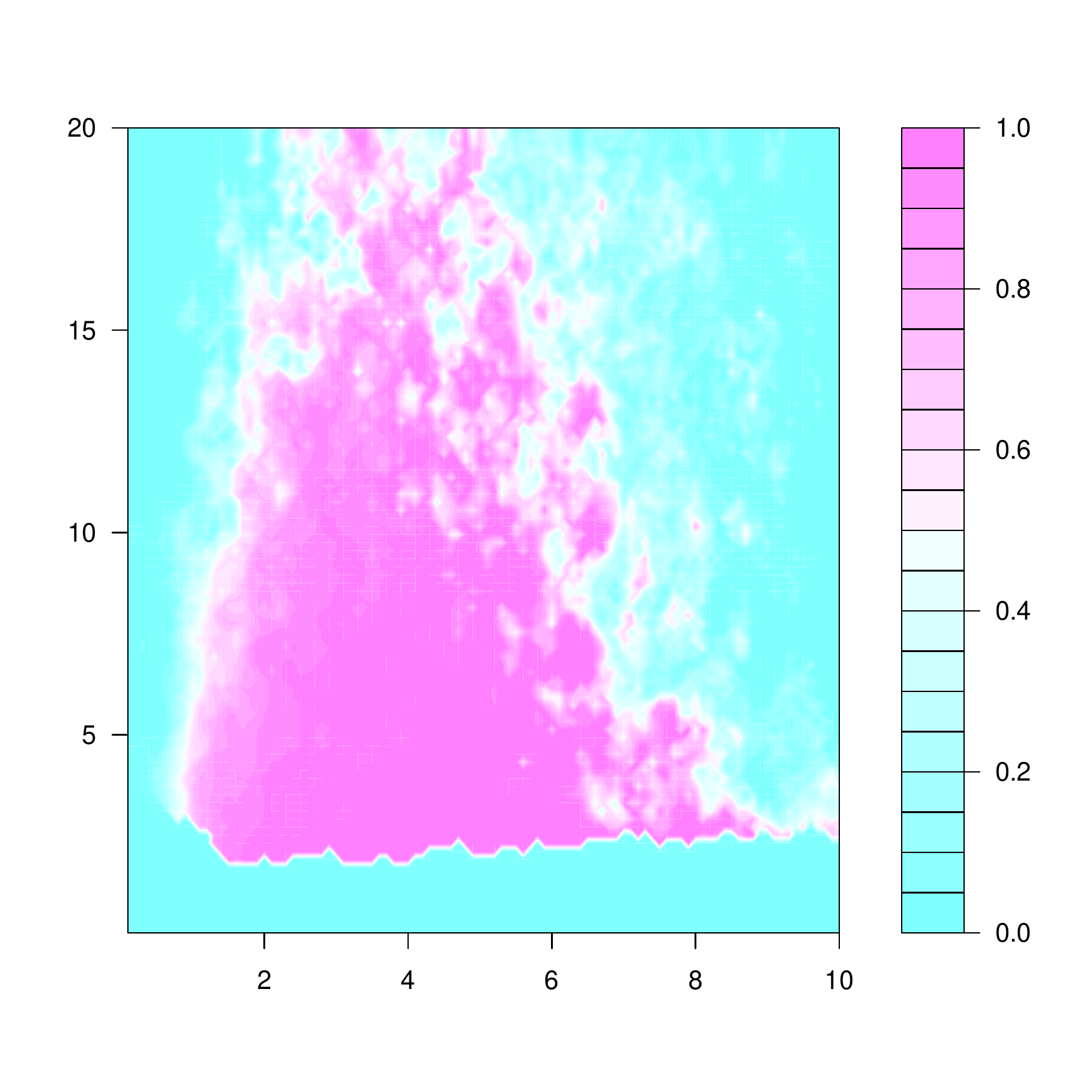}
  \caption{}
\end{subfigure}
\begin{subfigure}{.45\textwidth}
      \includegraphics[width=1\textwidth]{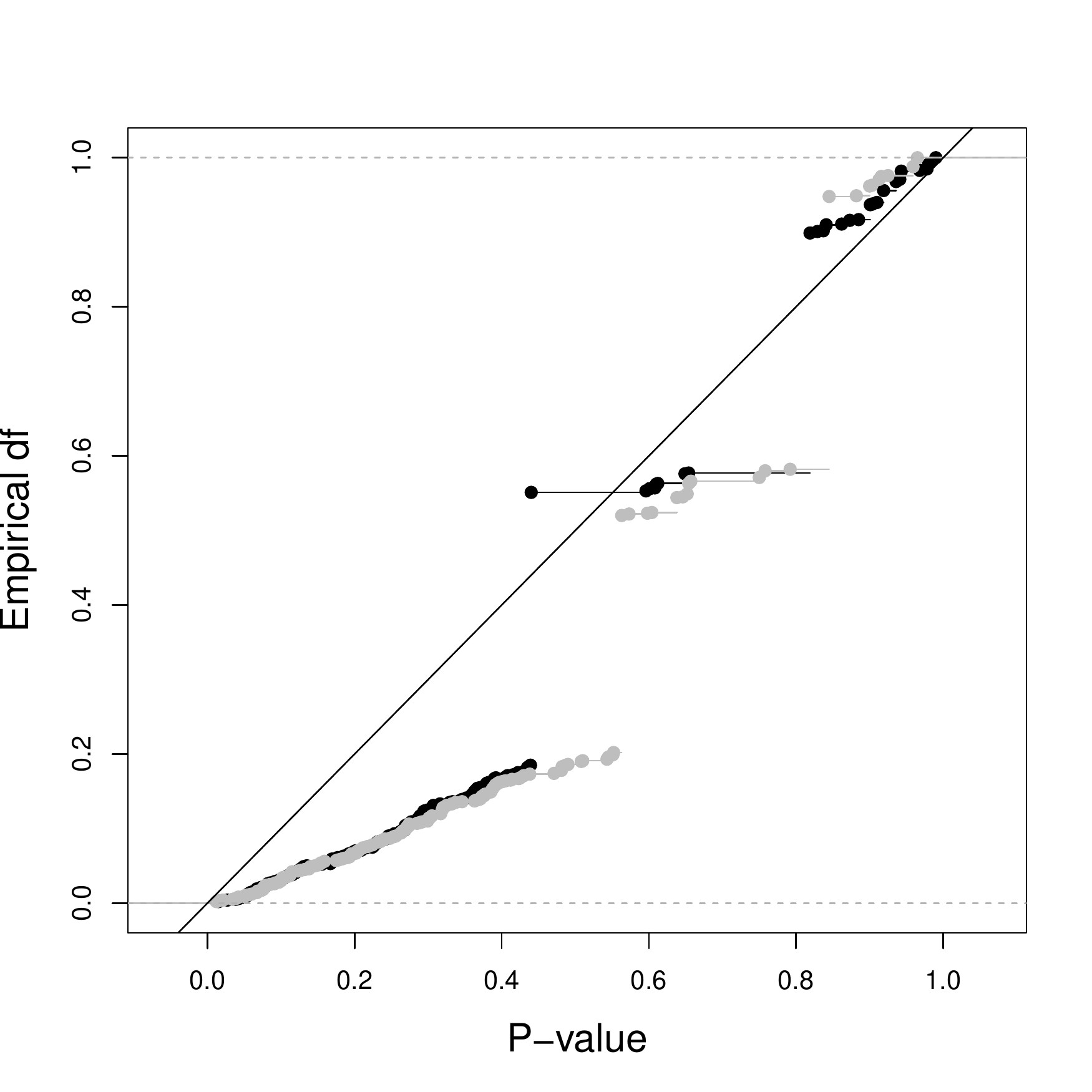}
  \caption{}
\end{subfigure}
  \caption{\label{Bioassay} (a) Estimated degree of weak informativity at $\gamma = 0.05$ for the bioassay example. (b) The plot points in grey and black represent estimated distribution of the conflict p-value using VSUGS-adjusted approach and by direct simulation from the prior predictive distribution.  Both estimation uses the alternative prior $\sigma_0 = \sigma_1 = 4$.    }
\end{figure}

We consider a bioassay example from \citet{Racine1986} which is also analysed in \citet{Gelman2008}, \citet{Evans2011} and \citet{Nott2015}. In this dataset, four groups of five animals were exposed to different level of doses ($x_i$) and the number of death ($y_i$) were recorded. Following \citet{Nott2015} and \citet{Evans2011}, we consider a logistic regression setup. It is assumed that the covariate has been transformed to log scale and centred and scaled as in \citet{Gelman2008}.  The model is $y_i \sim \mbox{Bin}(5, p_i)$ where $\textrm{logit} (p_i) = c_0 + c_1 x_i$. We assume that the priors for $c_0$ and $c_1$ are independent and follow Gaussian distributions with zero mean and variances $\sigma_0^2$ and $\sigma_1^2$ respectively. For our base prior, we consider $\sigma_0 = 10$ and $\sigma_1 = 2.5$. 

\citet{Evans2011} consider the exact sufficient statistics $(y_1,y_2,y_3,y_4)$ for analysis. \citet{Nott2015} consider using the posterior mode $(\hat{c}_0, \hat{c}_1)$ 
for a prior with $\sigma_0=\sigma_1=10$ as an approximation to the MLE but which unlike the MLE will exist even in degenerate cases.  
They consider the MLE for the dimension reduction that it brings
and as a generic choice applicable in situations where a non-trivial minimal sufficient statistic doesn't exist.  For the statistic $S$ used to define the conflict check in the definition
of weak informativity, they use a transformation of $(\hat{c}_0, \hat{c}_1)$ to the fitted probabilities $\hat{p}_2$ and $\hat{p}_3$ at $x_2$ and $x_3$ respectively. 
The reasons for this are discussed further in \citet{Nott2015}.
That is, our approximate sufficient statistic is $(\hat{p}_2,\hat{p}_3)$ where $\hat{p}_i = 1/(1+ \exp(-\hat{c}_0 + \hat{c}_1 x_i))$ for $ i = 2, 3$. 
Note that, because of the discreteness of the data, strictly the distribution of this statistic is also discrete but continuity may be used as a reasonable approximation when
the number of different possible values is large and we do this here.  Note also that in the kernel density estimation we ignore any boundary effects due to the bounded
support of the statistics.  

To use our methodology to investigate weak informativity with respect to the base prior in this example we proceed as follows.  First, we generate 400,000 values of $(\sigma_0,\sigma_1)$ from a pseudo prior which is uniform distribution on $[0.1,10] \times [0.1,20]$. We label these values as $\sigma^{(i)} = (\sigma^{(i)}_0,\sigma^{(i)}_1)$ for $i = 1,...,400,000$. For each $\sigma^{(i)}$, we generate $c^{(i)}_0$ and $c^{(i)}_1$ from their respectively prior distribution and then compute the probability $(p^{(i)}_1,p^{(i)}_2,p^{(i)}_3,p^{(i)}_4)$. We use these probabilities to generate $(y^{(i)}_1,y^{(i)}_2,y^{(i)}_3,y^{(i)}_4)$ from their respective binomial distribution. 

Let $\hat{p}^{(i)} = (\hat{p}^{(i)}_2, \hat{p}^{(i)}_3)^T$ and $K(\sigma^{(i)},\sigma^{(j)})$ and $K(\sigma^{(i)},\sigma^{(j)}) = \exp \{ - || \sigma^{(i)} - \sigma^{(j)}|| / 2\kappa^2 \}$, where $||.||$ is the Euclidean norm and $\kappa^2$ is the mean of the euclidean distance among 5000 random samples drawn from $\sigma^{(1)},...,\sigma^{(400,000)}$. \citet{Zhang2010} propose a similar choice of the kernel hyperparameter $\kappa^2$ and it is verified to be effective in their experimental analysis. We fit the matrix-variate Dirchlet process mixture model with $\hat{p}^{(i)}$ as our response vector and the basis functions $(1,K(\sigma^{(i)},\sigma^{(1)}),..,K(\sigma^{(i)},\sigma^{(N)}))$ as our covariates. We set $\alpha = 100$ and $T = 4$. 
For our prior, we set $a_{\tau} = 5$, $b_{\tau} = 0.5$, $a_i = 5$, $b_i = 0.5$, $S = I_2 + \frac{1}{2} \mathbf{1}_2 \mathbf{1}_2^T$, $\nu = 3$ and $M_1,...,M_T, C_1,..,C_T$ as zero matrices. 

We first run Algorithm \ref{MVSUGS} to fit the model. For each $\lambda_{\sigma}$ on a $100 \times 100$ regular grid on $[0.1,10] \times [0.1,20]$, we take 1000 nearest neighbours from the set $\{\sigma^{(i)}\}_{i = 1,...,400,000}$ using the \texttt{knnsearch} function in \texttt{matlab}. Then, we use these 1000 nearest neighbour to estimate the corresponding $(\hat{p}_2, \hat{p}_3)^T$ using the regression adjusted approach proposed in Section \ref{RegVSUGS}.  
Note that if we were to generate $1000$ samples directly for the prior predictive for each of our $10,000$ grid points directly this would increase the number of prior predictive
simulations and the computational effort by an order of magnitude.
The collection of $\lambda_{\sigma}$ covers the support of the hyperprior and each $\lambda_{\sigma}$ is considered as an alternative prior for comparison with the base prior.  We follow the suggestion from \citet{Evans2011} to measure the degree of weak informativity of an alternative prior. That is, we let $p_{\gamma}$ be the $\gamma \%$ quantile of the conflict $p$-value distribution for the base prior. Under the alternative prior, let $q_{\gamma}$ be the probability of a conflict p-value which is less than or equal to $p_{\gamma}$. The degree of weak informativity $\zeta_{\gamma}$ is defined as 
\begin{align}
\zeta_{\gamma} =\left\{\begin{array}{ll}
  0 & q_{\gamma} > p_{\gamma} \\
 1 - q_{\gamma}/p_{\gamma} & q_{\gamma} \leq p_{\gamma} 
\end{array}\right. \label{weakinform}
 \end{align}
Figure \ref{Bioassay}(a) plots the degree of weak informativity $\zeta_{\gamma}$ for all $\lambda_{\sigma}$ when $\gamma = 0.05$. We observe that the plot is very similar to Figure 2 in \citet{Nott2015}. From the plot, it seems that $\sigma_0 = \sigma_1 = 4$ is a suitable choice for a weakly informative prior. This conclusion agrees with the parameter choice of  \citet{Nott2015}.  As mentioned earlier, the approximate regression calculations are simply screening calculations where high accuracy is not needed since the quality of
the final answer can be checked.  Figure \ref{Bioassay}(b) shows a comparison of the estimated distribution of the conflict $p$-values based on regression (grey) compared to one based on direct simulation from the prior predictive at the finally chosen $\lambda$ (black).  In the lower tail, which is what matters for declaring the existence of any conflict and defining
weak informativity, the two distributions
agree very well, and the approximate regression calculations have successfully allowed us to identify a suitable weakly informative prior.
Note that standard procedures such as Gibbs sampler proposed in \citet{Zhang2010} are not suitable for use in this application with the matrix-variate DP prior model 
as we need to generate a large number of data points from
the hyperprior to obtain good estimates of the prior predictive distributions $p(S|\lambda)$ and so a method is needed that is able to handle large datasets.

\section{Empirical comparisons of predictive performance}

In this section, we focus on the predictive performance of the various approaches when fitting a model.  We consider the Gibbs sampler \citep{Zhang2010} and three different versions of the variational procedure. The first variational approach, which we call VB (Stick breaking), is the method of \citet{blei+j06}.   The second variational method, which we call VB (P\'{o}lya urn) is the batch variational method discussed in Section \ref{VBbatch}. In the third variational approach, we consider the VSUGS approach discussed in Section \ref{matrixVSUGS}. 

 In our experiment, the predictors are standardized to have zero mean and unit variance with respect to the training set. We fixed the number of iterations for the Gibbs sampler to be 25000, of which the first 15000 will be discarded as burn in. For the remaining 10000 iterations of the Gibbs sampler, we retain every 10th realization of the parameters. We also fixed the number of iterations for the variational approach using batch updates to 100. For the online variational approach, we first initialize the variational parameters of the assignment variables on a relatively small number of data points using the batch update and then run Algorithm \ref{MVSUGS}. 

We measure the performance of the various approaches by considering their root mean square error (RMSE) and mean absolute percentage error (MAPE). Let $\tilde{y}_1,..,\tilde{y}_m$ be our target response values and $\hat{y}_1,..,\hat{y}_m$ be their respective fitted value.  The error indicators are defined as
\[
RMSE = \sqrt{\frac{1}{m} \sum_{i=1}^m (\tilde{y}_i - \hat{y}_i  )^2}
\]
and
\[
MAPE = \frac{1}{m} \sum_{i=1}^m \left| \frac{\tilde{y}_i - \hat{y}_i}{\tilde{y}_i} \right|.
\]
Similar performance measures are used in \citet{Zhang2010}.

In our analysis, we compute the RMSE and MAPE both in-sample (for the training set) and out-of-sample (for the test set). 
Although we are mostly interested in out-of-sample predictive performance, looking at in-sample measures of fit can also be useful here where we are comparing several
computational approximations for the same posterior;  measures of in-sample fit can be revealing about differences in the quality of posterior approximation even
if out-of-sample predictive performance is similar for the different methods.  
For the RMSE and MAPE of in-sample predictions, they are constructed as follows. For the Gibbs sampler, at each retained MCMC realization $r$, $r=1,\dots ,1000$, we have a cluster allocation for each $y_1,...,y_n$. Suppose let's say that for a particular data point $y_i$, the allocation at the $r$ realization is the $j$th component. Then our corresponding regression coefficient estimate would then be $\beta^*_{i,r} = \beta_{j,r}^T$, where $\beta_{1,r},..,\beta_{T,r}$ are the $r$th MCMC coefficient matrix realizations. Then, our in-sample fitted value for $y_i$ is estimated as 
\[
\left \{ \sum_{r=1}^R \frac{\beta^*_{i,r}}{R} \right\}^T E_i. 
\]
For all variational procedures, for each $y_i$, we have the corresponding variational posterior probability of the assignments as well as the posterior mean. We use the weighted sum of these coefficient matrices according to the posterior probability to compute an in-sample fitted value.

Out of sample fitted values are obtained from the posterior predictive distributions of the respective procedures. For the Gibbs sampler, we will use the predictive distribution from equation (7) in \citet{Zhang2010}. For the VSUGS procedures, we use the posterior predictive mean to fit each $y_i$. Details of the posterior predictive mean can be found in Section \ref{Post_pred}.  For prediction accuracy of the test set, we also consider the adjusted VSUGS approach, which is to use the VSUGS with the regression-type adjustment in Section \ref{RegVSUGS}.

\subsection{Energy Data}

In this example, we consider the energy efficiency data created by \citet{Tsanas2012}. This dataset is available at  http://archive.ics.uci.edu/ml/datasets.html. \citet{Tsanas2012} studied the effect of eight input variables (relative compactness, surface area, wall area, roof area, overall height, orientation, glazing area, glazing area distribution) on two output variables, namely heating load (HL) and cooling load (CL), of residential buildings. The dataset contains 768 instances. We randomly select 100 data points as the test set and use the remaining 668 data points as the training set. 

The settings for fitting our model are as follows. We set the number of basis function, $N$, at 200 and use the same kernel discussed in Section \ref{bioassay}. Our choice of $\alpha$ is 3, which is set by rounding off the average of the 10000 iterations of $\alpha$ retained from the Gibbs sampler. For comparison of performance, we run both the Gibbs sampler and the variational procedures with a fixed $\alpha = 3$. For the hyperparameters of the priors, we set $a_{\tau} = 5$, $b_{\tau} = 0.5$, $a_i = 20$, $b_i = 0.5$, $S = I_2 + \frac{1}{2} \mathbf{1}_2 \mathbf{1}_2^T$, $\nu = 3$ and $M_1,...,M_T, C_1,..,C_T$ as zero matrices. We also set the maximum number of possible components as $T = 10$. For the variational approach with batch update, we initialize each assignment variables randomly to one of the $T$ components and set its variational probability to one. For the matrix VSUGS approach, we initialize the variational parameters for the first 200 assignment variables using the VB(P\'{o}lya urn). 

\begin{table}[!ht]
\caption{In-sample accuracy of various approaches for the energy efficiency data.  \label{Energy_In} }
{\small
\centering
{
\begin{tabular}{|c|cc|c|cc|c|c|} \hline
& \multicolumn{3}{|c|}{RMSE} & \multicolumn{3}{|c|}{MAPE}  &  \\ \hline
& $y_1$ & $y_2$ & Mean& $y_1$& $y_2$& Mean & Time (mins)   \\ \hline
Gibbs Sampler  & 0.0474 & 0.0450 & 0.0462 & 0.0605 & 0.0667 & 0.0636 & 141 \\ 
VB (Stick Breaking) & 0.1853 & 0.2150 & 0.2001 & 0.2163 & 0.2469 & 0.2163 & 18  \\
VB (P\'{o}lya urn) & 0.1579 & 0.1562 & 0.1570 & 0.1954 & 0.2806 & 0.2380 & 18   \\  
Matrix VSUGS & 0.3387 & 0.3160 & 0.3273 & 0.4527 & 0.5176 & 0.4851 & 3 \\ \hline
\end{tabular}
}
}
\end{table}


\begin{table}[!ht]
\caption{Prediction accuracy of various approaches for test set for the energy efficiency data  \label{Energy_Out} }
{\small
\centering
{
\begin{tabular}{|c|cc|c|cc|c|c|} \hline
& \multicolumn{3}{|c|}{RMSE} & \multicolumn{3}{|c|}{MAPE}   \\ \hline
& $y_1$ & $y_2$ & Mean& $y_1$& $y_2$& Mean    \\ \hline
Gibbs Sampler  &  0.5736 & 0.5881 & 0.5809 & 0.7427 & 0.8055 & 0.7741 \\ 
VB (Stick Breaking) &  0.3930 & 0.4911 & 0.4421 & 0.5731 & 0.8347 & 0.7039 \\
VB (P\'{o}lya urn) & 0.4005 & 0.5000  &0.4503  &0.4784  &0.5038  &0.4911 \\
Matrix VSUGS & 0.4038 & 0.4882 & 0.4460 & 0.5364 & 0.6140 & 0.5752 \\
Adjusted VSUGS & 0.2558 & 0.3329 & 0.2943 & 0.3399 & 0.4687 & 0.4043 \\ \hline
\end{tabular}

}
}
\end{table}


Measures of in-sample fit and computation times are presented in Table 2.  
The two batch VB methods have similar in-sample fits, but the Gibbs sampling and VSUGS approaches produce quite different results.  Table 3 considers 
out-of-sample predictive accuracy.  
Table \ref{Energy_Out} shows that all three variational approaches without the regression-type adjustment perform similarly. The RMSE and MAPE of the Gibbs sampler is higher than the rest. 
Table \ref{Energy_Out} also shows that using regression-type adjustment significantly improves prediction accuracy with respect to both RMSE and MAPE. 

\subsection{Robot Arm Data}

In this subsection, we analyse the performance of our proposed algorithm on the robot arm data. This dataset, 
available from \texttt{www.gaussianprocess.org/gpml/data}, 
relates to an inverse dynamics problem for a seven degrees-of-freedom SARCOS anthropomorphic robot arm. The
dataset has 21 covariates and 7 responses and has training and test sets of sizes 44448 and 4449 respectively. The 21 covariates consist of 7 joint positions, 7 joint velocities and 7 joint accelerations and the 7 responses consist of 7 joint torques.  

We follow the same procedure used in the energy dataset for our settings. The value of $\alpha$ is set to 12. For the hyperparameters of the prior, we set $a_{\tau} = 5$, $b_{\tau} = 0.5$, $a_i = 20$, $b_i = 0.5$, $S = I_7 + \frac{1}{7} \mathbf{1}_7 \mathbf{1}_7^T$, $\nu = 8$ and $M_1,...,M_T, C_1,..,C_T$ as zero matrices. We also set the maximum number of possible components as $T = 10$. For the matrix VSUGS approach, we initialize the variational parameters of the first 500 data points using VB (P\'{o}lya urn).

\begin{table}[!ht]
\caption{In-sample accuracy of various approaches for the robot arm data.  \label{Robot_In} }
{\footnotesize
\centering
{
\begin{tabular}{|c|c|ccccccc|c|} \hline
& Method & $y_1$ & $y_2$ & $y_3$ & $y_4$ & $y_6$ & $y_6$ & $y_7$ & Mean    \\ \hline
\multirow{ 4}{*}{RMSE} & Gibbs Sampler  &  0.1439 & 0.1338 & 0.1154 & 0.0854 & 0.1615 & 0.1615 & 0.1698 & 0.1304  \\  
& VB (Stick Breaking) & 0.2058 & 0.1827 & 0.1492 & 0.1512 & 0.1917 & 0.2086 & 0.1521  & 0.1773 \\
& VB (P\'{o}lya urn)  & 0.1887 & 0.1976 & 0.1704 & 0.1287 & 0.2042 & 0.2121 & 0.1403 & 0.1774  \\  
& Matrix VSUGS & 0.5297 & 0.4665 & 0.3894 & 0.4125 & 0.4531 & 0.4630 & 0.4097 & 0.4463  \\ \hline
\multirow{ 4}{*}{MAPE} & Gibbs Sampler  &   0.6156 &  0.4852 & 0.5554 & 0.5607 & 0.5902 & 0.8916 & 0.4822 & 0.5973  \\  
& VB (Stick Breaking) & 0.8257 & 0.5441  & 0.5941  & 0.7593  & 0.6417  & 1.0223  & 0.5984  & 0.7122 \\
& VB (P\'{o}lya urn)  & 0.6466 & 0.5814 & 0.6625 & 0.6585 & 0.6605 & 0.9861 & 0.6677 & 0.6948   \\  
& Matrix VSUGS & 1.4872 & 1.0410& 1.1798& 1.5815& 1.1536 &1.5569& 1.3610 & 1.3373 \\ \hline
\end{tabular}
}
}
\end{table}


\begin{table}[!ht]
\caption{Prediction accuracy of various approaches for the robot arm data.  \label{Robot_Out} }
{\footnotesize
\centering
{
\def\arraystretch{1.1}
\begin{tabular}{|c|c|ccccccc|c|} \hline
& Method & $y_1$ & $y_2$ & $y_3$ & $y_4$ & $y_6$ & $y_6$ & $y_7$ & Mean    \\ \hline
\multirow{ 5}{*}{RMSE} & Gibbs Sampler  &  0.4099 & 0.3636 & 0.3404 & 0.3638  & 0.3598 & 0.3867 & 0.3395 & 0.3662 \\ 
& VB (Stick Breaking) & 0.5524 & 0.4909 & 0.4547 & 0.4225 & 0.5227 & 0.4653 & 0.4037 & 0.4732 \\
& VB (P\'{o}lya urn) & 0.4323 & 0.4195 & 0.3983 & 0.3851 & 0.4410 & 0.3916 & 0.3762 & 0.4063   \\  
& Matrix VSUGS & 0.4198 & 0.3650 & 0.3375 & 0.3684 & 0.3737 & 0.3867 & 0.3490  & 0.3714 \\ 
& Adjusted VSUGS & 0.3513 & 0.3105 & 0.2869 & 0.2910 & 0.3249 & 0.3479 & 0.2741 & 0.3124 \\   \cline{1-10} 
\multirow{ 5}{*}{MAPE} & Gibbs Sampler  & 1.0327 & 0.5987 & 1.5229 & 1.4137 & 1.0385 & 0.7330 &  1.2885 & 1.0897 \\
& VB (Stick Breaking) & 1.3711 & 0.7539 & 1.9556 & 1.0920 & 1.2136 & 1.1278 & 1.3825 & 1.2709 \\
& VB (P\'{o}lya urn)  & 0.7587 & 0.6397 & 0.9774 & 0.9263 & 1.4217 & 1.1021 & 1.1947 & 1.0029   \\  
& Matrix VSUGS & 0.9679 & 0.6600  & 1.8472 & 1.4600 & 1.2228 & 0.8428 & 1.4971  & 1.2140 \\
& Adjusted VSUGS & 0.7813  & 0.6097 & 1.7723 & 0.6627 & 1.1116 & 0.8206 & 1.0809 & 0.9770 \\ \hline
\end{tabular}
}
}
\end{table}

\begin{table}[!ht]
\caption{Prediction accuracy of various approaches using full training set for robot arm data.  \label{Robot_Full} }
{\footnotesize
\centering
{
\begin{tabular}{|c|c|ccccccc|c|} \hline
Method&  & $y_1$ & $y_2$ & $y_3$ & $y_4$ & $y_6$ & $y_6$ & $y_7$ & Mean    \\ \hline
\multirow{2}{*}{RMSE} & Matrix VSUGS  & 0.3930 & 0.3453 & 0.3065 & 0.3071 & 0.3402 & 0.3330 & 0.2923 & 0.3311 \\
 & Adjusted VSUGS & 0.2505 & 0.2500 & 0.2315 & 0.1626 & 0.2541 & 0.2515 & 0.1785 & 0.2255 \\ \hline
\multirow{2}{*}{MAPE} & Matrix VSUGS & 0.8212 & 0.5288 & 1.9353 & 1.0970 & 0.9967 & 1.0329 & 0.7880 & 1.0285 \\
 & Adjusted VSUGS & 0.7204 & 0.4575 & 1.8183 & 0.7178 & 0.8687 & 0.7910 & 0.7220 & 0.8708 \\ \hline
\end{tabular}
}
}
\end{table}

The performance of the in-sample fits is presented in Table \ref{Robot_In}, for fitting to a subsample of size $2000$. Corresponding out-of-sample fits for test set subsample of size 500 are shown in Table \ref{Robot_Out}.  
We also initialize the variational parameters of the first 2000 data points using VB (P\'{o}lya urn) and then run Algorithm \ref{MVSUGS} on the full training dataset of size 44448. As the dataset is large, we do not use the Gibbs sampler or the variational procedures using batch updates for the full dataset. The results of out-of-sample predictive performance are presented in Table \ref{Robot_Full}. It is clear that there is a significant prediction accuracy improvement in terms of RMSE and MAPE when using the adjusted VSUGS. 

\subsection{Computation Times}

\begin{table}[!ht]
\centering
\caption{Computation time for Robot Arm Data.  \label{Robot_time} }
{\footnotesize
{
\begin{tabular}{|c|c|c|} \hline
Method  & \multicolumn{2}{c|}{Time (minutes)}  \\ \hline
Sample size &  n = 2000  & n = 44448      \\ \hline
 Gibbs Sampler & 2422 & -  \\ 
 VB (Stick Breaking) & 368 & - \\
 VB (P\'{o}lya urn) & 365 & -   \\  
 Matrix VSUGS & 93 & 584 \\ \hline
\end{tabular}
}
}
\end{table}

All the algorithms are run on a \texttt{Mac 3.2Ghz i5} Quad core processor with code written in \texttt{matlab}. For the energy efficiency dataset, as reflected in Table \ref{Energy_In}, the computation time for Gibbs sampler, VB (Stick breaking), VB (P\'{o}lya urn) and matrix VSUGS require approximately 141, 18, 18 and 3 minutes respectively.  For the robot arm dataset, Table \ref{Robot_time} shows that the Gibbs sampler, VB (Stick breaking), VB (P\'{o}lya urn) and matrix VSUGS require 2422, 368, 365 and 93 minutes respectively. The amount of time required for matrix VSUGS to run the full dataset is 584 minutes, which is significantly shorter than the amount of time it takes for the Gibbs sampler to run on a much smaller dataset. In fact, two third of the computation time is spent on initializing the first 2000 data points.

\section{Discussion}

In this article, we study variational computational methods for fitting the matrix-variate Dirichlet process mixture model of \citet{Zhang2010}, extending the VSUGS approach 
by \citet{zhang+nyj14} for Dirichlet process mixtures of normal densities.  The method we develop is computationally efficient and especially useful as an alternative to MCMC for analysis 
of medium to large datasets.   
In order to increase prediction accuracy, we also propose a regression-type adjustment for improving predictive inference. The adjustment approach is shown to be useful in
several real applications.

\section*{Appendix A - Variational Batch update}

In the derivation of the update for each block $\gamma$ say we will write simply $E_q(\cdot)$ for the expectation with respect to the current variational posterior distribution $q$ with $\gamma$ integrated out and will not denote the dependence on the block explicitly in the notation.  The meaning will be clear from the context.  We consider the mean field update for $(\beta,\Sigma)$ first.  We have
$$q(\beta,\Sigma) \propto \exp\left(E_q\left(\sum_{i=1}^n\sum_{j=1}^T I(\delta_i=j)\log p(y_i|\beta_j^TE_i,\tau\Sigma)+\sum_{j=1}^T \log p(\beta_j|\Omega,\Sigma)+\log p(\Sigma)\right)\right).$$
Apart from constant terms not depending on $\beta,\Sigma$ we have
$$E_q(\log p(\Sigma)) =-\frac{\nu+m+1}{2}\log |\Sigma|-\frac{1}{2}\tr(S\Sigma^{-1})$$
$$E_q(\log p(\beta_j|\Omega,\Sigma)) =-\frac{N+1}{2}\log |\Sigma|-\frac{1}{2}\tr\left((E_q(\Omega^{-1})+C_j)(\beta_j-M_j)\Sigma^{-1}(\beta_j-M_j)^T\right)$$
$$E_q(\sum_{i=1}^n\sum_{j=1}^T I(\delta_i=j)\log p(y_i|\beta_j^TE_i,\tau\Sigma))=\sum_{i=1}^n \sum_{j=1}^T q_{ij}E_q(\log p(y_i|\beta_j^TE_i,\tau\Sigma)$$
where (again apart from terms not depending on $\beta,\Sigma$)
\begin{align*}
E_q(\log p(y_i|\beta_j^TE_i,\tau\Sigma))&=-\frac{1}{2}\log|\Sigma|-\frac{1}{2}E_q(\tau^{-1})(y_i-\beta_j^TE_i)^T\Sigma^{-1}(y_i-\beta_j^T E_i).
\end{align*}
This gives (again up to an additive constant)
\begin{align*}
 \log q(\beta,\Sigma) &= -\frac{(n+N+2)m+\nu+1}{2}\log |\Sigma|-\frac{1}{2}\tr(S\Sigma^{-1})  \\
& \qquad -\frac{1}{2}\sum_{j=1}^T \tr\left((E_q(\Omega^{-1})+C_j)(\beta_j-M_j)\Sigma^{-1}(\beta_j-M_j)^T\right) \\
 & \qquad -\frac{1}{2}E_q(\tau^{-1})\sum_{i=1}^n\sum_{j=1}^T q_{ij}(y_i-\beta_j^TE_i)^T\Sigma^{-1}(y_i-\beta_j^TE_i).
\end{align*}
To simplify this, write 
$$ \hat{\beta}_j=V_j^{-1}\left((E_q(\Omega^{-1})+C_j)M_j+E_q(\tau^{-1})\sum_iq_{ij}E_iy_i^T\right) $$
where
$$V_j=\left(E_q(\Omega^{-1})+C_j+E_q(\tau^{-1})\sum_i q_{ij}E_iE_i^T\right) $$
and observe that
\begin{align*}
& \tr\left((E_q(\Omega^{-1})+C_j)(\beta_j-M_j)\Sigma^{-1}(\beta_j-M_j)^T\right) \\
&\qquad + E_q(\tau^{-1})\sum_{i=1}^n q_{ij}(y_i-\beta_j^TE_i)^T\Sigma^{-1}(y_i-\beta_j^TE_i) \\
& = \tr\left((\beta_j-M_j)^T(E_q(\Omega^{-1})+C_j)(\beta_j-M_j)\Sigma^{-1}\right) \\
&\qquad + \tr\left(E_q(\tau^{-1})\sum_{i=1}^n q_{ij}(y_i-\beta_j^T E_i)(y_i-\beta_j^TE_i)^T\Sigma^{-1}\right) \\
& = \tr((\beta_j-\hat{\beta}_j)^T V_j (\beta_j-\hat{\beta}_j)\Sigma^{-1})+ \tr\left(M_j^T(E_q(\Omega^{-1})+C_j)M_j\Sigma^{-1}\right) \\
&\qquad +\tr\left(E_q(\tau^{-1})\sum_{i=1}^n q_{ij}y_iy_i^T \Sigma^{-1}\right) - \tr\left(\hat{\beta}_j^T V_j \hat{\beta}_j\Sigma^{-1}\right) .
\end{align*}
This means that up to an additive constant
\begin{align*}
& \log q(\beta,\Sigma)= \\
& -\frac{(n+N+2)m+\nu+1}{2}\log |\Sigma|-\frac{1}{2}\sum_{j=1}^T\tr((\beta_j-\hat{\beta}_j)^TV_j(\beta_j-\hat{\beta}_j)\Sigma^{-1}) \\
& -\frac{1}{2}\tr\left(\left(S+\sum_{j=1}^T \left(  E_q(\tau^{-1})\sum_{i=1}^n q_{ij}y_iy_i^T+M_j^T
(E_q(\Omega^{-1})+C_j)M_j-\hat{\beta}_j^T V_j \hat{\beta}_j\right)\right)\Sigma^{-1}\right).
\end{align*}
Then we recognize the form of $q(\beta,\Sigma)$ as being $q(\beta,\Sigma)=q(\Sigma)q(\beta|\Sigma)$ where
$q(\Sigma)$ is inverse Wishart, 
$$ IW\left(\nu+n,S+\sum_{j=1}^T \left(E_q(\tau^{-1})\sum_{i=1}^n q_{ij}y_iy_i^T+M_j^T
(E_q(\Omega^{-1})+C_j)M_j-\hat{\beta}_j^T V_j \hat{\beta}_j\right) \right) $$
and $q(\beta|\Sigma)=\prod_{j=1}^T q(\beta_j|\Sigma)$ with $q(\beta_j|\Sigma)=N_{N+1,m}(\hat{\beta}_j,V_j^{-1}\otimes \Sigma)$.

Next, let's consider the mean field update for $q(\tau)$.  We have
$$q(\tau)\propto \exp(E_q(\sum_{i=1}^n \sum_{j=1}^T I(\delta_i=j)\log p(y_i|\beta_j^TE_i,\tau\Sigma)+\log p(\tau)).$$
Apart from additive constants
$$E_q(\log p(\tau))=\log p(\tau)=-(a_{\tau}+1)\log \tau -b_\tau/\tau$$
and
$$E_q\left(\sum_{i=1}^n\sum_{j=1}^T I(\delta_i=j)\log p(y_i|\beta_j^T E_i,\tau\Sigma)\right)=\sum_{i=1}^n\sum_{j=1}^T q_{ij}E_q(\log p(y_i|\beta_j^TE_i,\tau\Sigma)$$
where
$$E_q(\log p(y_i|\beta_j^TE_i,\tau\Sigma))=\frac{-nm}{2}\log \tau-\frac{1}{2\tau} E_q((y_i-\beta_j^T E_i)\Sigma^{-1}(y_i-\beta_j^T E_i)).$$
Next,
$$E_q((y_i-\beta_j^TE_i)^T\Sigma^{-1}(y_i-\beta_j^TE_i))=E_q(E_q((y_i-\beta_j^TE_i)^T\Sigma^{-1}(y_i-\beta_j^TE_i)|\Sigma)).$$
To evaluate the inner conditional expectation, we use the following Lemma (Guptar and Nagar, p. 60).   \\
\noindent
\begin{lemma}\label{EXAX}
Suppose that $X\sim N_{p,q}(M,\Delta\otimes \psi)$. Let $A$ be a $p\times p$ matrix.  Then 
$$E(X^TAX)=\tr(\Delta A^T)\psi+M^TAM$$
\end{lemma}
\noindent
Using Lemma \ref{EXAX} and some simple algebra we obtain
$$E_q((y_i-\beta_j^TE_i)^T\Sigma^{-1}(y_i-\beta_j^TE_i)|\Sigma)=(y_i-\hat{\beta}_j^T E_i)^T\Sigma^{-1}(y_i-\hat{\beta}_j^T E_i)+mE_i^TV_j^{-1} E_i.$$
Hence
$$E_q((y_i-\beta_j^TE_i)^T\Sigma^{-1}(y_i-\beta_j^T E_i))=(y_i-\hat{\beta}_j^T E_i)^TE_q(\Sigma^{-1})(y_i-\hat{\beta}_j^TE_i)+mE_i^TV_j^{-1}E_i.$$
So apart from additive constants
\begin{align*}
\log q(\tau) & =-(a_\tau+\frac{nm}{2}+1)\log \tau \\
& \qquad -\frac{1}{\tau} \left\{ b_\tau+\frac{1}{2}\sum_{i=1}^n\sum_{j=1}^T q_{ij}((y_i-\hat{\beta}_j^TE_i)^TE_q(\Sigma^{-1})(y_i-\hat{\beta}_j^TE_i)+mE_i^TV_j^{-1}E_i)) \right\}.
\end{align*}
Hence we recognize that $q(\tau)$ is inverse gamma, 
$$ IG(a_\tau+\frac{nm}{2},b_\tau+\frac{1}{2}\sum_{i=1}^n\sum_{j=1}^T q_{ij}((y_i-\hat{\beta}_j^TE_i)^TE_q(\Sigma^{-1})(y_i-\hat{\beta}_j^TE_i)+mE_i^TV_j^{-1}E_i)). $$
Next we consider the variational update for $\omega$.  We have that apart from additive constants
\begin{align*}
\log q(\omega) &=-\sum_{i=1}^{N+1}(a_i+1)\log\omega_i-\sum_{i=1}^{N+1} \frac{b_i}{\omega_i}-\frac{mT}{2}\sum_{i=1}^{N+1}\log \omega_i \\
& \qquad -\frac{1}{2}\sum_{i=1}^{N+1}\frac{1}{\omega_i}\sum_{j=1}^T E_q((\beta_{j,i}-M_{j,i})^T\Sigma^{-1}(\beta_{j,i}-M_{j,i}))
\end{align*}
where $\beta_{j,i}$ and $M_{j,i}$ denote the $i$th rows of $\beta_j$ and $M_j$ respectively.  Writing
$$E_q((\beta_{j,i}-M_{j,i})^T\Sigma^{-1}(\beta_{j,i}))=E_q(E_q((\beta_{j,i}-M_{j,i})^T\Sigma^{-1}(\beta_{j,i})|\Sigma))$$
and noting that $\beta_{j,i}|\Sigma \sim N(\hat{\beta}_{j,i},\omega_{ij}' \Sigma)$ where $\hat{\beta}_{j,i}$ is the $i$th row of $\hat{\beta}_j$ and
$\omega_{ij}'$ is the $i$th diagonal element of $V_j^{-1}$ we have
\begin{align*}
 E_q((\beta_{j,i}-M_{j,i})^T\Sigma^{-1}(\beta_{j,i})) &= (\hat{\beta}_{j,i}-M_{j,i})^T\Sigma^{-1}(\hat{\beta}_{j,i}-M_{j,i})+\tr(\Sigma^{-1}\omega_{ij}'\Sigma) \\
& = (\hat{\beta}_{j,i}-M_{j,i})^TE_q(\Sigma^{-1})(\hat{\beta}_{j,i}-M_{j,i})+\omega_{ij}'m.
\end{align*}
So apart from additive constants
\begin{align*}
\log q(\omega) &=-\sum_{i=1}^{N+1}(a_i+mT/2+1)\log \omega_i \\
&\qquad -\sum_{i=1}^{N+1}\frac{1}{\omega_i}(b_i+\frac{1}{2}\sum_{j=1}^T \left((\hat{\beta}_{j,i}-M_{j,i})^TE_q(\Sigma^{-1})(\hat{\beta}_{j,i}-M_{j,i})+\omega_{ij}'m\right)
\end{align*}
and so we recognize that $q(\omega)=\prod_{i=1}^{N+1}q(\omega_i)$
where $q(\omega_i)$ is inverse gamma, 
\begin{align}\label{batchomega}
 IG(a_i+mT/2,b_i+\frac{1}{2}\sum_{j=1}^T \left((\hat{\beta}_{j,i}-M_{j,i}) E_q(\Sigma^{-1})(\hat{\beta}_{j,i}-M_{j,i})^T+m \omega_{ij}'\right)). 
\end{align}

\section*{Appendix B - Sequential update for $\delta_i$}

\noindent
To help evaluate the integral we use the following lemma. 

\noindent
\begin{lemma}\label{MatrixLemma}
\begin{enumerate}[(a)]
\item Let $Z$, $V$, $W$ and $C$ be matrices with $Z$ $s\times t$, $V>0$ $s\times s$, $W>0$ $t\times t$ and $C$ $s\times t$.  Then
\begin{align*}
\int \exp(\tr(-\frac{1}{2}(V^{-1}ZW^{-1}Z^T-2V^{-1}CW^{-1}Z^T))dZ \\
=(2\pi)^{st/2}|V|^{t/2}|W|^{s/2}\exp(\tr(\frac{1}{2}V^{-1}CW^{-1}C^T)).
\end{align*}
\item Suppose $\psi$, $A$ are matrices with $\psi >0$ 
$m\times m$ 
and $A>0$ $m\times m$. Let $a>0$ be a constant.  Then
\begin{align*}
& \int | \psi |^{-(a+m+1)/2}
\exp(\tr(-\frac{1}{2}A\psi^{-1})) d\psi=|A|^{-a/2} 2^{am/2} \Gamma_m (a/2)
\end{align*}
where $\Gamma_m(x)$ is the multivariate gamma function, $\Gamma_m(x)=\pi^{m(m-1)/4}\prod_{i=1}^m \Gamma(x+(1-i)/2)$.  \\
\end{enumerate}
\end{lemma}
\begin{proof}
Part a) of the lemma follows easily from the fact that the integrand is an unnormalized matrix normal distribution.  Using the fact that the integral of the corresponding normalized density is one and rearranging gives the result.  Part b) follows in a similar way by noting that the integrand is an unnormalized inverse Wishart density.  
\end{proof}

Using Lemma \ref{MatrixLemma} to help evaluate the integral we get
\begin{align*}
& \int p(y_i|\beta_j^TE_i,\tau\Sigma)q_{i-1}(\tau)q_{i-1}(\beta_j|\Sigma)q_{i-1}(\Sigma)d\beta_jd\Sigma d\tau  \\
& =  \int\int(2\pi\tau)^{-m/2}|\Sigma|^{-1/2}q_{i-1}(\tau)q_{i-1}(\Sigma)\exp\left\{ -\frac{1}{2\tau}y_i^T\Sigma^{-1}y_i \right\} (2\pi)^{-(N+1)m/2}| \\
&\qquad\qquad \Omegajione|^{-m/2}|\Sigma|^{-(N+1)/2} \exp(\tr(-\frac{1}{2} (\Omegajione)^{-1} \bhatjione \Sigma^{-1} (\bhatjione)^T)) \\
&\qquad \int\exp \left\{ \tr \left[ -\frac{1}{2}\left((\frac{1}{\tau}E_iE_i^T+(\Omegajione)^{-1})\beta_j\Sigma^{-1}\beta_j^T \right.\right.\right. \\
&\qquad\qquad \left.\left.\left. -2\left(\frac{1}{\tau}E_iE_i^T+ (\Omegajione)^{-1}\right)\bar{\beta}_{j}\Sigma^{-1}\beta_j^T\right) \right] \right\} d\beta_jd\Sigma d\tau
\end{align*}
where $\bar{\beta}_j= \left(\frac{1}{\tau}E_iE_i^T+ (\Omegajione)^{-1}\right)^{-1} \left(\frac{1}{\tau}E_iy_i^T+  (\Omegajione)^{-1} \bhatjione \right)$.
Next, use part (a) of Lemma \ref{MatrixLemma} and the explicit form for $q_{i-1}(\Sigma)$ to get
\begin{align*}
& \int (2\pi\tau)^{-m/2}q_{i-1}(\tau)|\Omegajione|^{-m/2}|\frac{1}{\tau}E_iE_i^T +(\Omegajione)^{-1}|^{-m/2}\frac{|\Sione|^{\nuione/2}}{2^{\nuione m/2}\Gamma_m(\nuione/2)} \\
& \int |\Sigma|^{-(\nuione+m+2)/2}\exp(\tr(-\frac{1}{2}(\Sione +\frac{1}{\tau}y_iy_i^T+(\bhatjione)^T (\Omegajione)^{-1} \bhatjione \\
& \qquad -\bar{\beta}_j^T(\frac{1}{\tau}E_iE_i^T+
(\Omegajione)^{-1})\bar{\beta}_j)\Sigma^{-1}))d\Sigma d\tau.
\end{align*}
Using part (b) of Lemma \ref{MatrixLemma}, this is equal to
\begin{align*}
& \int (2\pi \tau)^{-\frac{m}{2}}q_{i-1}(\tau)|\Omegajione|^{-\frac{m}{2}}|\frac{1}{\tau}E_iE_i^T+(\Omegajione)^{-1}|^{-\frac{m}{2}} \\
&\qquad \frac{|\Sione|^{\nuione/2}\Gamma_m((\nuione+1)/2)}{2^{\nuione \frac{m}{2}}\Gamma_m(\nuione/2)} 2^{(\nuione+1)\frac{m}{2}}\\
&\qquad   \left |\Sione +\frac{1}{\tau}y_iy_i^T+ (\bhatjione)^T (\Omegajione)^{-1}\bhatjione -\bar{\beta}_j^T(\frac{1}{\tau}E_iE_i^T+ (\Omegajione)^{-1})\bar{\beta}_j \right|^{-\frac{\nuione+1}{2}}d\tau.
\end{align*}

\section*{Appendix C - Posterior predictive distribution}

We now evaluate the integral in \eqref{predictive}. Following the same steps to get \eqref{qdelta}, we have
\begin{align*}
& \int p( y_0 | y_{1:n}, \tau, \Sigma, \beta_j, \delta_0 = j) q_n(\tau ) q_n(\beta_j) q_n(\Sigma)    d\beta_{j} d\tau d\Sigma  \\
& =  \left(\frac{2\pi}{\mutaun} \right)^{-m/2}|\Omegajn|^{-m/2}|\mutaun E_0 E_0^T+ (\Omegajn)^{-1}|^{-m/2}\frac{2^{\frac{(\nun+1)m}{2}} \Gamma_m\left(\frac{\nun+1}{2}\right)}{2^{\nun m/2}\Gamma_m(\nun/2)}  |\Sn|^{\nun/2}\\
&  |S^{(n)}+\mutaun y_0y_0^T+(\hat{\beta}^{(n)}_{j})^T(\Omegajn)^{-1} \hat{\beta}^{(n)}_{j}-\tilde{\beta}_j^T( \mutaun E_0E_0^T+ (\Omegajn)^{-1})\tilde{\beta}_j|^{-\frac{\nun+1}{2}}.
\end{align*}
where $\tilde{\beta}_j=(\mutaun  E_0E_0^T+(\Omegajn)^{-1})^{-1}(\mutaun E_0y_0^T+(\Omegajn)^{-1} \bhatjn)$. 
We now show that after simplification of the term inside the determinant we obtain the multivariate $t$-distribution.
We require here the matrix determinant lemma, which states the following.
\begin{lemma}\label{Matrix_det}
For any invertible matrix $A$ and vector $u$ and $v$ we have
\[
|A + uv^T| = (1  + v^T A^{-1} u) |A|.
\]
\end{lemma}
Letting $\Lambda = (\mutaun  E_0E_0^T+(\Omegajn)^{-1})$, it is clear that
\begin{align*}
& \tilde{\beta}_j^T \left(\mutaun  E_0E_0^T+ (\Omegajn)^{-1} \right)\tilde{\beta}_j \\
&= \left(\mutaun  E_0y_0^T+  (\Omegajn)^{-1} \bhatjn \right)^T \Lambda^{-1} \Lambda \Lambda^{-1}  \left(\mutaun  E_0y_0^T+  (\Omegajn)^{-1} \bhatjn \right) \\
&= (\mutaun)^2 y_0 E_0^T \Lambda^{-1} E_0 y_0^T +\mutaun  y_0 E_0^T \Lambda^{-1}  (\Omegajn)^{-1}\bhatjn  \\
&\qquad + \mutaun (\bhatjn)^T (\Omegajn)^T \Lambda^{-1} E_0 y_0^T + (\bhatjn)^T (\Omegajn)^T \Lambda^{-1} \Omegajn \bhatjn.
\end{align*}
This gives us
\begin{align*}
& \biggl|\Sn + \mutaun y_0y_0^T+ (\bhatjn)^T (\Omegajn)^{-1}\bhatjn-\tilde{\beta}_j^T(\mutaun E_0E_0^T+ (\Omegajn)^{-1})\tilde{\beta}_j \biggr|^{-\frac{\nun+1}{2}} \\
&= \biggl|\Sn+ \mutaun y_0y_0^T+ (\bhatjn)^T (\Omegajn)^{-1} \bhatjn-  (\mutaun)^2 y_0 E_0^T \Lambda^{-1} E_0 y_0^T \\
&\qquad  - \mutaun y_0 E_0^T \Lambda^{-1}  \Omegajn \bhatjn - \mutaun (\bhatjn)^T (\Omegajn)^T \Lambda^{-1} E_0 y_0^T \\
&\qquad   -  (\bhatjn)^T (\Omegajn)^T \Lambda^{-1} \Omegajn \bhatjn \biggr|^{-\frac{\nun+1}{2}} \\
&= \biggl|\Sn + \mutaun y_0 (I -  \mutaun E_0^T \Lambda^{-1} E_0 ) y_0^T  - \mutaun y_0 E_0^T \Lambda^{-1} \Omegajn \bhatjn    \\
&\qquad  - \mutaun (\bhatjn)^T (\Omegajn)^T \Lambda^{-1} E_0 y_0^T -  (\bhatjn)^T (\Omegajn)^T \Lambda^{-1} \Omegajn \bhatjn \\
&\qquad + (\bhatjn)^T (\Omegajn)^{-1}\bhatjn \biggr|^{-\frac{\nun +1}{2}}.
\end{align*}
Let $S_* = \Sn  -  (\bhatjn)^T (\Omegajn)^T \Lambda^{-1} \Omegajn \bhatjn + (\bhatjn)^T (\Omegajn)^{-1} \bhatjn$ (where we suppress dependence on $j$ in the notation).  
Using Lemma \ref{Matrix_det}, we have
\begin{align*}
\left|1 + \mutaun y_0^T S_*^{-1} y_0 (1 -  \mutaun E_0^T \Lambda^{-1} E_0 ) - 2\mutaun E_0^T \Lambda^{-1} \Omegajn \bhatjn S_*^{-1} y_0  \right|^{-\frac{\nun +1}{2}} \left|S_* \right|^{-\frac{\nun +1}{2}}.
\end{align*}
As the term inside the absolute value is a quadratic function in $y_0$, the predictive distribution is a multivariate $t$-distribution. Letting $\mathcal{A} = \mutaun (1 -  \mutaun E_0^T \Lambda^{-1} E_0 )    S_*^{-1}$, $\mathcal{B} =  -  2 \mutaun  S_*^{-1} (\bhatjn)^T \Omegajn \Lambda^{-1}  E_0$ (again suppressing dependence on $j$ in the notation), we have
\begin{align*}
&\left |1 + \mutaun y_0^T S_*^{-1} y_0 (1 -  \mutaun E_0^T \Lambda^{-1} E_0 ) -   2\mutaun y_0^T S_*^{-1} (\bhatjn)^T \Omegajn \Lambda^{-1}  E_0   \right|  \\
&=  \left|1 + (y_0 + \frac{1}{2} \mathcal{A}^{-1} \mathcal{B} )^T \mathcal{A}   (y_0 + \frac{1}{2} \mathcal{A}^{-1} \mathcal{B})  -\frac{1}{4} \mathcal{B}^T \mathcal{A}^{-1} \mathcal{B} \right| \\
&= \left|1 + (y_0 + \frac{1}{2} \mathcal{A}^{-1} \mathcal{B})^T \mathcal{A}  (1  -\frac{1}{4} \mathcal{B}^T \mathcal{A}^{-1} \mathcal{B} )^{-1}  (y_0 + \frac{1}{2} \mathcal{A}^{-1} \mathcal{B}) \right| \left|1 -\frac{1}{4} \mathcal{B}^T \mathcal{A}^{-1} \mathcal{B} \right|.
\end{align*}
Thus, the $j$th component in the expression for the posterior predictive density \eqref{predictive} is a multivariate $t$-density with location $-\frac{1}{2} \mathcal{A}^{-1} \mathcal{B}$ and variance $\frac{1}{\nun - m - 1} \left\{ \mathcal{A} ( 1 -\frac{1}{4} \mathcal{B}^T \mathcal{A}^{-1} \mathcal{B})^{-1} \right\}^{-1}$ .  Hence \eqref{predictive} is approximately
a mixture of multivariate $t$-densities.

\section*{Appendix D - Variational Lower bound}

Next we compute the variational lower bound on $\log p(y_{1:n})$, which is defined as
\[
\mathcal{L} (q) = E_q \left\{ \log  \left[ p \left(y_{1:n},  \beta_{1:T}, \delta_{1:n}, \tau, \omega_{1:(N+1)}, \Sigma \right) \right] \right\} - E_q \left\{ \log  \left[ q \left( \beta_{1:T}, \delta_{1:n}, \tau, \omega_{1:(N+1)}, \Sigma \right) \right] \right\}.
\]
Similar to the approach used to calculate the lower bound for the normal mixture model in \citet{zhang+nyj14}, we approximate $\mathcal{L} (q)$ recursively. Let $\theta_t = (\beta_t, \tau, \omega_{1:N+1}, \Sigma)$ and using  $r_{ij}$ as an approximation to $p(\delta_i | y_{i:i-1})$, we have
\begin{align*}
\mathcal{L} (q) &= E_{q_i} \left\{ \log(q_{i-1} (\beta_{1:T}, \tau, \omega_{1:(N+1)}, \Sigma) p(\delta_i | y_{i:i-1}) p(y_i | \delta_i , \theta)) \right\} - E_{q_i} \left\{\log ( q_{i} (\theta) q_i (\delta_i) )\right\} \\
&= E_{q_i} \biggl\{ \log (q_{i-1} (\tau)) + \log(q_{i-1} (\Sigma)) + \sum_{j=1}^{N+1} \log(q_{i-1} (\omega_j) + \sum_{t=1}^T \log(q_{i-1} (\beta_t)) + \log(r_{ij}) +   \\
&\qquad   \log(y_i | \theta_{ \delta_i }) \biggr\} -  E_{q_i} \left\{ q_i (\tau) +  \log(q_{i} (\Sigma)) + \sum_{j=1}^{N+1} \log(q_{i} (\omega_j) \right. \\
&\qquad \left. + \sum_{t=1}^T \log(q_{i} (\beta_t)) + \log(q_i (\delta_i)) \right\}.  
\end{align*}
Suppressing the expectations with respect to $q_i$, evaluating the terms involving $\tau$ and $\omega_j$ gives us
\begin{align*}
E_{q_i} \left\{ \log (q_{i-1} (\tau)) \right\} - E_{q_i} \left\{ \log (q_{i} (\tau)) \right\}   &=  (\ataui - \atauione) \psi (\ataui) - \log (\Gamma (\ataui)) + \log (\Gamma (\atauione)) \\
& \qquad + \atauione ( \log (\btauione) - \log (\btaui)) + \ataui \frac{\btaui - \btauione}{\btaui},
\end{align*}
\begin{align*}
E_{q_i} \left\{ \log (q_{i-1} (\omega_j)) \right\} - E_{q_i} \left\{ \log (q_{i} (\omega_j)) \right\}   &=  (\aji - \ajione) \psi (\aji) - \log (\Gamma (\aji)) + \log (\Gamma (\ajione)) \\
& \qquad + \ajione ( \log (\bjione) - \log (\bji)) + \aji \frac{\bji - \bjione}{\bji}.
\end{align*}
Since $q_i (\Sigma)$ follows an inverse Wishart distribution with degrees of freedom $v^{(i)}$ and scale matrix $S^{(i)}$, we obtain
\begin{align*}
E_{q_i} \left\{ \log(q_{i-1} (\Sigma)) \right\} &= \frac{\nuione}{2} \log|\Sione| - \frac{\nuione m}{2} \log(2) - \log \Gamma_m (\frac{\nuione}{2})   \\
&\qquad  - \frac{1}{2} tr \left\{ \Sione E_{q_i} (\Sigma^{-1}) \right\} - \frac{\nuione + m + 1}{2} E_{q_i} \left( \log |\Sigma| \right) \\
&= \frac{\nuione}{2} \log|\Sione| - \frac{\nuione m}{2} \log(2) - \log \Gamma_m (\frac{\nuione}{2})  \\
&\qquad - \frac{1}{2} tr \left\{ \Sione (\Si)^{-1} \nui  \right\} \\
&\qquad  - \frac{\nuione + m + 1}{2} \left\{ - \psi_m \left(\frac{\nui}{2} \right) - m \log (2) + \log |\Si| \right\} \\
\end{align*}
\begin{align*}
E_{q_i} \left\{ \log(q_{i} (\Sigma)) \right\} &= \frac{\nui}{2} \log|\Si| - \frac{\nui m}{2} \log(2) - \log \Gamma_m (\frac{\nui}{2}) - \frac{1}{2} tr \left\{ \Si E_{q_i} (\Sigma^{-1}) \right\} \\
&\qquad  - \frac{\nui + m + 1}{2} E_{q_i} \left( \log |\Sigma| \right)  \\
&= \frac{\nui}{2} \log|\Si| - \frac{\nui m}{2} \log(2) - \log \Gamma_m (\frac{\nui}{2})  - \frac{1}{2} tr \left\{ \Si (\Si)^{-1} \nui  \right\} \\
&\qquad  - \frac{\nui + m + 1}{2} \left\{ - \psi_m \left(\frac{\nui}{2} \right) - m \log (2) + \log |\Si| \right\}. \\
\end{align*}
Next, observe that if $Z$ is matrix-variate normal distribution $N_{s,t} (C, V \otimes W)$ , then $E(Z^T B Z) = W \tr(V B^T) + C^T B C$. Moreover, for $A$ a $r \times s$ matrix, we have $AX \sim N_{r,t} (AC, AVA^T \otimes W)$ and $E(AX \Lambda A^T X^T) = AVA^T \tr(\Lambda W) + (AC)^T \Lambda  AC$. Therefore, we have
\begin{align*}
E_{q_i} \left\{  \log(q_{i-1} (\beta_j)) \right\} &= - \frac{(N+1)m}{2} \log (2\pi) - \frac{N+1}{2} E_{q_i} ( \log|\Sigma| ) - \frac{m}{2} \log |(V^{(i-1)}_j)^{-1}| \\
&\qquad - E_{q_i} \left\{ \frac{1}{2} \tr \left[\Sigma^{-1} (\beta_j - \bhatjione)^T V^{(i-1)}_j (\beta_j - \bhatjione)  \right] \right\} \\
&= - \frac{(N+1)m}{2} \log (2\pi) - \frac{N+1}{2} E_{q_i} ( \log|\Sigma| ) - \frac{m}{2} \log |(V^{(i-1)}_j)^{-1}| \\
&\qquad -  \frac{1}{2} \tr \biggl[ \nui (\Si)^{-1} (\bhatji - \bhatjione)^T V^{(i-1)}_j (\bhatji - \bhatjione)   \\
&\qquad + E_{q_i} \left( \Sigma^{-1} \Sigma \tr( (V_j^{(i)})^{-1} V_j^{(i-1)} ) \right)  \biggr] \\
&= - \frac{(N+1)m}{2} \log (2\pi) - \frac{N+1}{2} E_{q_i} ( \log|\Sigma| ) - \frac{m}{2} \log |(V^{(i-1)}_j)^{-1}| \\
&\qquad -  \frac{1}{2} \tr \biggl[ \nui (\Si)^{-1} (\bhatji - \bhatjione)^T V^{(i-1)}_j (\bhatji - \bhatjione) \\
& \qquad  + I_m \tr( (V_j^{(i)})^{-1} V_j^{(i-1)} ) \biggr] 
\end{align*}
\begin{align*}
E_{q_i} \left\{  \log(q_{i} (\beta_j)) \right\} &= - \frac{(N+1)m}{2} \log (2\pi) - \frac{N+1}{2} E_{q_i} ( \log|\Sigma| ) - \frac{m}{2} \log |(V^{(i)}_j)^{-1}| \\
&\qquad - E_{q_i} \left\{ \frac{1}{2} \tr \left[(V^{(i)}_j) (\beta_j - \bhatji)^T \Sigma^{-1} (\beta_j - \bhatji)  \right] \right\} \\
&= - \frac{(N+1)m}{2} \log (2\pi) - \frac{N+1}{2} E_{q_i} ( \log|\Sigma| )  \\
&\qquad - \frac{m}{2} \log |(V^{(i)}_j)^{-1}| - \frac{1}{2} m(N+1)
\end{align*}
\begin{align*}
E_{q_i} \left\{ \log p(y_i | \theta_{\delta_i}) \right\} &= \sum_{j=1}^T \hat{q}_i (\delta_i = j) E_{q_i} \left\{ \log(y_i | \theta_{j} ) \right\}
\end{align*}
where
\begin{align*}
E_{q_i} \left\{ \log p(y_i | \theta_{j} ) \right\} &= -\frac{m}{2} \log (2 \pi) - \frac{m}{2} E_{q_i} \left\{ \log \tau \right\} - \frac{1}{2} E_{q_i} \left\{ \log |\Sigma| \right\} \\
&\qquad -\frac{1}{2} E_{q_i} \left\{ \frac{1}{\tau} \right\} E_{q_i} \left\{ y_i^T\Sigma^{-1} y_i - 2 E_i^T \beta_j \Sigma^{-1} y_i + E_i^T \beta_j \Sigma^{-1} \beta_j^T E_i \right\} \\ 
&= -\frac{m}{2} \log (2 \pi) -  \frac{m}{2}  \left\{  \log (\btaui) - \psi(\ataui) \right\} \\
&\qquad -  \frac{1}{2}  \left\{ - \psi_m \left(\frac{\nui}{2} \right) - m \log (2) + \log |\Si| \right\}  \\
&\qquad -\frac{1}{2} \frac{\ataui}{\btaui}  \left\{ \nui y_i^T (\Si)^{-1} y_i - 2 \nui  E_i^T \hat{\beta}^{(i)}_j (\Si)^{-1} y_i \right. \\
&\qquad \left. + m  E_i^T \Omega_j^{(i)} E_i  + (E_i^T \hat{\beta}^{(i)}_j)^T (\nui (\Si)^{-1}) (E_i^T \hat{\beta}^{(i)}_j) \right\}. 
\end{align*}
Finally, we have
\begin{align*}
E_{q_i} \left\{ \log(q_i (\delta_i) \right\} &= \sum_{j=1}^{T} \hat{q}_i (\delta_i = j) \log (\hat{q}_i (\delta_i = j)).
\end{align*}

\bibliographystyle{chicago}
\bibliography{matrixvariatedp}

\begin{thebibliography}{}

\bibitem[\protect\citeauthoryear{Attias}{Attias}{2000}]{attias00}
Attias, H. (2000).
\newblock A variational {B}ayesian framework for graphical models.
\newblock In {\em Advances in Neural Information Processing Systems 12}, pp.\
  209--215. MIT Press.

\bibitem[\protect\citeauthoryear{Beaumont, Zhang, and Balding}{Beaumont
  et~al.}{2002}]{beaumont+zb02}
Beaumont, M.~A., W.~Zhang, and D.~J. Balding (2002).
\newblock Approximate {Bayesian} computation in population genetics.
\newblock {\em Genetics\/}~{\em 162}, 2025--2035.

\bibitem[\protect\citeauthoryear{Blei and Frazier}{Blei and
  Frazier}{2011}]{blei+f11}
Blei, D.~M. and P.~I. Frazier (2011).
\newblock Distance dependent {C}hinese restaurant processes.
\newblock {\em J. Mach. Learn. Res.\/}~{\em 12}, 2461--2488.

\bibitem[\protect\citeauthoryear{Blei and Jordan}{Blei and
  Jordan}{2006}]{blei+j06}
Blei, D.~M. and M.~I. Jordan (2006).
\newblock Variational inference for {D}irichlet process mixtures.
\newblock {\em Bayesian Anal.\/}~{\em 1}, 121--143.

\bibitem[\protect\citeauthoryear{Blum and Tran}{Blum and Tran}{2010}]{blum+t10}
Blum, M.~G. and V.~C. Tran (2010).
\newblock {HIV} with contact tracing: {A} case study in approximate {B}ayesian
  computation.
\newblock {\em Biostatistics\/}~{\em 11\/}(4), 644--660.

\bibitem[\protect\citeauthoryear{Blum}{Blum}{2010}]{blum10}
Blum, M. G.~B. (2010).
\newblock Approximate {B}ayesian computation: {A} nonparametric perspective.
\newblock {\em Journal of the American Statistical Association\/}~{\em
  105\/}(491), 1178--1187.

\bibitem[\protect\citeauthoryear{Blum and Fran\c{c}ois}{Blum and
  Fran\c{c}ois}{2010}]{blum+f10}
Blum, M. G.~B. and O.~Fran\c{c}ois (2010).
\newblock Non-linear regression models for approximate {Bayesian} computation.
\newblock {\em Statistics and Computing\/}~{\em 20}, 63--75.

\bibitem[\protect\citeauthoryear{Box}{Box}{1980}]{Box1980}
Box, G. (1980).
\newblock {Sampling and Bayes' inference in scientific modelling and robustness
  (with discussion)}.
\newblock {\em Journal of the Royal Statistical Society, Series A\/}~{\em 143},
  383--430.

\bibitem[\protect\citeauthoryear{Bryant and Sudderth}{Bryant and
  Sudderth}{2012}]{bryant+s12}
Bryant, M. and E.~B. Sudderth (2012).
\newblock Truly nonparametric online variational inference for hierarchical
  {D}irichlet processes.
\newblock In {\em Advances in Neural Information Processing Systems 25}, pp.\
  2708--2716.

\bibitem[\protect\citeauthoryear{Caron, Davy, and Doucet}{Caron
  et~al.}{2007}]{caron+dd07}
Caron, F., M.~Davy, and A.~Doucet (2007).
\newblock Generalized {P}\'{o}lya urn for time-varying {D}irichlet process
  mixtures.
\newblock In R.~Parr and L.~C. van~der Gaag (Eds.), {\em UAI}, pp.\  33--40.
  AUAI Press.

\bibitem[\protect\citeauthoryear{Chen, Rao, Buntine, and Teh}{Chen
  et~al.}{2013}]{chen+rbt13}
Chen, C., V.~Rao, W.~L. Buntine, and Y.~W. Teh (2013).
\newblock Dependent normalized random measures.
\newblock In {\em ICML (3)}, Volume~28 of {\em JMLR Proceedings}, pp.\
  969--977.

\bibitem[\protect\citeauthoryear{De~Iorio, M\"{u}ller, Rosner, and
  MacEachern}{De~Iorio et~al.}{2004}]{deiorio+mrm04}
De~Iorio, M., P.~M\"{u}ller, G.~L. Rosner, and S.~N. MacEachern (2004).
\newblock An {ANOVA} model for dependent random measures.
\newblock {\em Journal of the American Statistical Association\/}~{\em 99},
  205--215.

\bibitem[\protect\citeauthoryear{Dunson and Park}{Dunson and
  Park}{2008}]{dunson+p08}
Dunson, D.~B. and J.-H. Park (2008).
\newblock Kernel stick-breaking processes.
\newblock {\em Biometrika\/}~{\em 95}, 307--323.

\bibitem[\protect\citeauthoryear{Evans and Jang}{Evans and
  Jang}{2011}]{Evans2011}
Evans, M. and G.~H. Jang (2011).
\newblock Weak informativity and the information in one prior relative to
  another.
\newblock {\em Statist. Sci.\/}~{\em 26\/}(3), 423--439.

\bibitem[\protect\citeauthoryear{Evans and Moshonov}{Evans and
  Moshonov}{2006}]{Evans2006}
Evans, M. and H.~Moshonov (2006).
\newblock Checking for prior-data conflict.
\newblock {\em Bayesian Anal.\/}~{\em 1\/}(4), 893--914.

\bibitem[\protect\citeauthoryear{Foti and Williamson}{Foti and
  Williamson}{2015}]{foti+w15}
Foti, N. and S.~Williamson (2015).
\newblock A survey of non-exchangeable priors for {B}ayesian nonparametric
  models.
\newblock {\em IEEE Transactions on Pattern Analysis and Machine
  Intelligence\/}~{\em 37}, 359--371.

\bibitem[\protect\citeauthoryear{Gelfand, Kottas, and MacEachern}{Gelfand
  et~al.}{2005}]{gelfand+km05}
Gelfand, A.~E., A.~Kottas, and S.~N. MacEachern (2005).
\newblock Bayesian nonparametric spatial modeling with {D}irichlet process
  mixing.
\newblock {\em Journal of the American Statistical Association\/}~{\em 100},
  1021--1035.

\bibitem[\protect\citeauthoryear{Gelman}{Gelman}{2006}]{Gelman2006}
Gelman, A. (2006).
\newblock Prior distributions for variance parameters in hierarchical models.
\newblock {\em Bayesian Anal.\/}~{\em 1\/}(3), 515--533.

\bibitem[\protect\citeauthoryear{Gelman, Jakulin, Pittau, and Su}{Gelman
  et~al.}{2008}]{Gelman2008}
Gelman, A., A.~Jakulin, M.~G. Pittau, and Y.-S. Su (2008).
\newblock A weakly informative default prior distribution for logistic and
  other regression models.
\newblock {\em Ann. Appl. Stat.\/}~{\em 2\/}(4), 1360--1383.

\bibitem[\protect\citeauthoryear{Griffin and Steel}{Griffin and
  Steel}{2006}]{griffin+s06}
Griffin, J.~E. and M.~F.~J. Steel (2006).
\newblock Order-based dependent {D}irichlet processes.
\newblock {\em Journal of the American Statistical Association\/}~{\em
  101\/}(473), 179--194.

\bibitem[\protect\citeauthoryear{Hoffman, Blei, Wang, and Paisley}{Hoffman
  et~al.}{2013}]{hoffman+bwp13}
Hoffman, M.~D., D.~M. Blei, C.~Wang, and J.~Paisley (2013).
\newblock Stochastic variational inference.
\newblock {\em J. Mach. Learn. Res.\/}~{\em 14}, 1303--1347.

\bibitem[\protect\citeauthoryear{Jordan, Ghahramani, Jaakkola, and Saul}{Jordan
  et~al.}{1999}]{Jordan1999}
Jordan, M.~I., Z.~Ghahramani, T.~S. Jaakkola, and L.~K. Saul (1999).
\newblock An introduction to variational methods for graphical models.
\newblock {\em Mach. Learn.\/}~{\em 37}, 183--233.

\bibitem[\protect\citeauthoryear{Kabisa, Dunson, and Morris}{Kabisa
  et~al.}{2016}]{tchumtchoua+d16}
Kabisa, S.~T., D.~B. Dunson, and J.~S. Morris (2016).
\newblock Online variational bayes inference for high-dimensional correlated
  data.
\newblock {\em Journal of Computational and Graphical Statistics\/}~{\em To
  appear}.

\bibitem[\protect\citeauthoryear{Kingman}{Kingman}{1967}]{kingman67}
Kingman, J. F.~C. (1967).
\newblock Completely random measures.
\newblock {\em Pacific Journal of Mathematics\/}~{\em 21}, 59--78.

\bibitem[\protect\citeauthoryear{Lijoi, Nipoti, and Pr\"{u}nster}{Lijoi
  et~al.}{2014}]{lijoi+np14}
Lijoi, A., B.~Nipoti, and I.~Pr\"{u}nster (2014).
\newblock Bayesian inference with dependent normalized completely random
  measures.
\newblock {\em Bernoulli\/}~{\em 20}, 1260--1291.

\bibitem[\protect\citeauthoryear{Lijoi and Pr\"{u}nster}{Lijoi and
  Pr\"{u}nster}{2010}]{lijoi+p10}
Lijoi, A. and I.~Pr\"{u}nster (2010).
\newblock Models beyond the {D}irichlet process.
\newblock In N.~L. Hjort, C.~Holmes, P.~M\"{u}ller, and S.~G. Walker (Eds.),
  {\em Bayesian Nonparametrics}, pp.\  80--136. Cambridge University Press.

\bibitem[\protect\citeauthoryear{Lin}{Lin}{2013}]{lin13}
Lin, D. (2013).
\newblock Online learning of nonparametric mixture models via sequential
  variational approximation.
\newblock In C.~Burges, L.~Bottou, M.~Welling, Z.~Ghahramani, and K.~Weinberger
  (Eds.), {\em Advances in Neural Information Processing Systems 26}, pp.\
  395--403. Curran Associates, Inc.

\bibitem[\protect\citeauthoryear{Luts, Broderick, and Wand}{Luts
  et~al.}{2014}]{luts+bw14}
Luts, J., T.~Broderick, and M.~P. Wand (2014).
\newblock Real-time semiparametric regression.
\newblock {\em Journal of Computational and Graphical Statistics\/}~{\em 23},
  589--615.

\bibitem[\protect\citeauthoryear{MacEachern}{MacEachern}{2000}]{maceachern00}
MacEachern, S.~N. (2000).
\newblock Dependent {D}irichlet processes.
\newblock Technical report, Department of Statistics, Ohio State University.

\bibitem[\protect\citeauthoryear{Nott, Drovandi, Mengersen, and Evans}{Nott
  et~al.}{2015}]{Nott2015}
Nott, D., C.~C. Drovandi, K.~Mengersen, and M.~Evans (2015).
\newblock Approximation of bayesian predictive p-values with regression abc.
\newblock Technical report, Queensland University of Technology.

\bibitem[\protect\citeauthoryear{Ormerod and Wand}{Ormerod and
  Wand}{2010}]{ormerod+w10}
Ormerod, J. and M.~Wand (2010).
\newblock Explaining variational approximations.
\newblock {\em The American Statistician\/}~{\em 64}, 140--153.

\bibitem[\protect\citeauthoryear{Racine, Grieve, Fl{\"u}hler, and Smith}{Racine
  et~al.}{1986}]{Racine1986}
Racine, A., A.~P. Grieve, H.~Fl{\"u}hler, and A.~F.~M. Smith (1986).
\newblock Bayesian methods in practice: {E}xperiences in the pharmaceutical
  industry.
\newblock {\em J. Roy. Statist. Soc. Ser. C\/}~{\em 35\/}(2), 93--150.

\bibitem[\protect\citeauthoryear{Rao and Teh}{Rao and Teh}{2009}]{rao+t09}
Rao, V.~A. and Y.~W. Teh (2009).
\newblock Spatial normalized gamma processes.
\newblock In Y.~Bengio, D.~Schuurmans, J.~Lafferty, C.~K.~I. Williams, and
  A.~Culotta (Eds.), {\em Advances in Neural Information Processing Systems
  22}, pp.\  1554--1562.

\bibitem[\protect\citeauthoryear{Sato}{Sato}{2001}]{sato01}
Sato, M.-A. (2001).
\newblock Online model selection based on the variational bayes.
\newblock {\em Neural Comput.\/}~{\em 13}, 1649--1681.

\bibitem[\protect\citeauthoryear{Teh, Jordan, Beal, and Blei}{Teh
  et~al.}{2006}]{teh+jbb06}
Teh, Y.~W., M.~I. Jordan, M.~J. Beal, and D.~M. Blei (2006).
\newblock Hierarchical {D}irichlet processes.
\newblock {\em Journal of the American Statistical Association\/}~{\em
  101\/}(476), 1566--1581.

\bibitem[\protect\citeauthoryear{Tsanas and Xifara}{Tsanas and
  Xifara}{2012}]{Tsanas2012}
Tsanas, A. and A.~Xifara (2012).
\newblock Accurate quantitative estimation of energy performance of residential
  buildings using statistical machine learning tools.
\newblock {\em Energy and Buildings\/}~{\em 49}, 560 -- 567.

\bibitem[\protect\citeauthoryear{Wang and Blei}{Wang and Blei}{2012}]{wang+b12}
Wang, C. and D.~M. Blei (2012).
\newblock Truncation-free online variational inference for {B}ayesian
  nonparametric models.
\newblock In F.~Pereira, C.~Burges, L.~Bottou, and K.~Weinberger (Eds.), {\em
  Advances in Neural Information Processing Systems 25}, pp.\  413--421. Curran
  Associates, Inc.

\bibitem[\protect\citeauthoryear{Wang, Paisley, and Blei}{Wang
  et~al.}{2011}]{wang+pb11}
Wang, C., J.~Paisley, and D.~M. Blei (2011).
\newblock {Online variational inference for the hierarchical Dirichlet
  process}.
\newblock In {\em Proc. of the 14th Int'l. Conf. on Artificial Intelligence and
  Statistics (AISTATS)}, Volume~15, pp.\  752--760.

\bibitem[\protect\citeauthoryear{Wang and Dunson}{Wang and
  Dunson}{2011}]{wang+d11}
Wang, L. and D.~B. Dunson (2011).
\newblock Fast {Bayesian} inference in {Dirichlet} process mixture models.
\newblock {\em Journal of Computational and Graphical Statistics\/}~{\em 20},
  196--216.

\bibitem[\protect\citeauthoryear{Waterhouse, Mackay, and Robinson}{Waterhouse
  et~al.}{1996}]{waterhouse+mr96}
Waterhouse, S., D.~Mackay, and T.~Robinson (1996).
\newblock Bayesian methods for mixture of experts.
\newblock In {\em Advances in Neural Information Processing Systems 8}, pp.\
  351--357. MIT Press.

\bibitem[\protect\citeauthoryear{Zhang, Nott, Yau, and Jasra}{Zhang
  et~al.}{2014}]{zhang+nyj14}
Zhang, X., D.~J. Nott, C.~Yau, and A.~Jasra (2014).
\newblock A sequential algorithm for fast fitting of {D}irichlet process
  mixture models.
\newblock {\em Journal of Computational and Graphical Statistics\/}~{\em 23},
  1143--1162.

\bibitem[\protect\citeauthoryear{Zhang, Dai, and Jordan}{Zhang
  et~al.}{2010}]{Zhang2010}
Zhang, Z., G.~Dai, and M.~Jordan (2010).
\newblock Matrix-variate {D}irichlet process mixture models.
\newblock {\em Proceedings of the Thirteenth Conference on Artificial
  Intelligence and Statistics (AISTATS)\/}~{\em 9}, 988--995.

\bibitem[\protect\citeauthoryear{Zhang, Wang, Dai, and Jordan}{Zhang
  et~al.}{2014}]{zhang+wdj14}
Zhang, Z., D.~Wang, G.~Dai, and M.~I. Jordan (2014).
\newblock Matrix-variate {D}irichlet process priors with applications.
\newblock {\em Bayesian Analysis\/}~{\em 9}, 259--286.

\end{thebibliography}

\end{document}